\pdfoutput=1
\documentclass[9pt, nocopyrightspace, numbers, preprint]{sigplanconf}

\usepackage{amsmath,amsthm,amsfonts,amssymb,amscd}
\usepackage{hyperref}
\usepackage{listings}
\usepackage{stmaryrd}
\usepackage{inconsolata}
\usepackage[inline]{aplcomments} 
\usepackage{xspace}
\usepackage{adjustbox}
\usepackage{fancyhdr}
\usepackage{mathrsfs}
\usepackage{xcolor}
\usepackage{graphicx}
\usepackage{dsfont}
\usepackage{tikz}
\usepackage{hyperref}

\newcommand{\lbr}{\llbracket}
\newcommand{\rbr}{\rrbracket}
\newcommand{\gvd}{\Gamma \vdash}
\newcommand{\teq}{\triangleq}

\newcommenter{sc}{1.0,0.4,1.0} 
\newcommenter{kw}{0.4,0.4,1.0} 
\newcommenter{ds}{1.0,0.4,0.4} 
\newcommenter{ac}{0.4,1.0,1.0}

\newcommand{\set}[1]{\{#1\}}                    
\newcommand{\setof}[2]{\{{#1}\mid{#2}\}}        
\newcommand{\join}{\bowtie}
\newcommand{\semijoin}{\ltimes}
\newcommand{\Integer}{\mathds{Z}}

\newcommand{\Bool}{\mathds{B}}
\newcommand{\true}{\texttt{true}}

\newcommand{\String}{String}

\newcommand{\Query}{Query}
\newcommand{\Pred}{Predicate}
\newcommand{\Expr}{Expression}
\newcommand{\Proj}{Projection}

\newcommand{\lambdaFn}[2]{\lambda \; {#1} \;.\;{#2}}
\renewcommand{\mathtt}[1]{{\textrm {\tt {#1}}}}
\newcommand{\D}[1]{\lbr{#1}\rbr}

\newcommand{\N}{\mathds{N}}
\newcommand{\denote}[1]{\lbr #1 \rbr}
\newcommand{\sep}{~|~}
\newcommand{\Context}{\Gamma}
\newcommand{\context}{g}
\newcommand{\Schema}{\texttt{Schema}}
\newcommand{\schema}{\sigma}
\newcommand{\tuple}{t}
\newcommand{\query}{q}
\newcommand{\proj}{p}
\newcommand{\pred}{b}
\newcommand{\expr}{e}

\newcommand{\nomial}{n}
\newcommand{\zero}{\textbf{0}}
\newcommand{\one}{\textbf{1}}

\newcommand{\leaf}[1]{\texttt{leaf} \; #1}
\newcommand{\node}[2]{\texttt{node} \; #1 \; #2}
\newcommand{\emptySchema}{\texttt{empty}}

\newcommand{\denoteQuery}[3]{\denote{#1 \vdash #2 : #3}}
\newcommand{\denoteExpr}[3]{\denote{#1 \vdash #2 : #3}}
\newcommand{\denoteProj}[3]{\denote{#1 : #2 \Rightarrow #3}}
\newcommand{\denotePred}[2]{\denote{#1 \vdash #2}}
\newcommand{\denoteTable}[1]{\denote{#1}}

\newcommand{\intro}{\lambda \; \context. \;}
\newcommand{\introt}{\lambda \; \tuple. \;}
\newcommand{\intros}{\lambda \; \context \; \tuple. \;}
\newcommand{\callQuery}[3]{\denoteQuery{#1}{#2}{#3} \; \context \; \tuple}
\newcommand{\callExpr}[3]{\denoteExpr{#1}{#2}{#3} \; \context}
\newcommand{\callPred}[2]{\denotePred{#1}{#2} \; \context}

\newcommand{\Tuple}{\texttt{Tuple}}
\newcommand{\Type}{\mathcal{U}}

\newcommand{\type}{\tau}

\newcommand{\Pair}[2]{\node{#1}{#2}}
\newcommand{\mkPair}[2]{(#1,#2)}
\newcommand{\fst}[1]{#1.1}
\newcommand{\snd}[1]{#1.2}

\newcommand{\Unit}{\texttt{Unit}}
\newcommand{\mkUnit}{\texttt{unit}}

\newcommand{\merely}[1]{\left\lVert#1\right\rVert}
\newcommand{\negate}[1]{#1 \rightarrow \zero}

\newcommand{\SELECT}[2]{\texttt{SELECT} \; #1 \; #2}
\newcommand{\WHERE}[2]{#1 \; \texttt{WHERE} \; #2}
\newcommand{\FROM}[1]{\texttt{FROM} \; #1}
\newcommand{\UNIONALL}[2]{#1 \; \texttt{UNION ALL} \; #2}
\newcommand{\EXCEPT}[2]{#1 \; \texttt{EXCEPT} \; #2}
\newcommand{\AND}[2]{#1 \; \texttt{AND} \; #2}
\newcommand{\OR}[2]{#1 \; \texttt{OR} \; #2}
\newcommand{\CastPred}[2]{\texttt{CASTPRED} \; #1 \; #2}
\newcommand{\CastExpr}[2]{\texttt{CASTEXPR} \; #1 \; #2}
\newcommand{\Compose}[2]{ #1 . \; #2}
\newcommand{\Duplicate}[2]{#1 , \; #2}
\newcommand{\DISTINCT}[1]{\texttt{DISTINCT} \; #1}
\newcommand{\NOT}[1]{\texttt{NOT} \; #1}
\newcommand{\EXISTS}[1]{\texttt{EXISTS} \; #1}

\newcommand{\Evaluate}[1]{\texttt{E2P}~#1}
\newcommand{\TRUE}{\texttt{TRUE}}
\newcommand{\FALSE}{\texttt{FALSE}}
\newcommand{\ProjL}{\texttt{Left}}
\newcommand{\ProjR}{\texttt{Right}}
\newcommand{\Empty}{\texttt{Empty}}
\newcommand{\Star}{\texttt{*}}
\newcommand{\Var}[1]{\texttt{P2E}~#1}

\newcommand{\SelectFrom}[2]{\SELECT{#1}{\FROM{#2}}}
\newcommand{\SelectFromWhere}[3]{\SelectFrom{#1}{\WHERE{#2}{#3}}}
\newcommand{\AS}{\; \texttt{AS} \;}

\newtheorem{theorem}{Theorem}[section]
\newtheorem{lemma}[theorem]{Lemma}

\newtheorem{definition}[theorem]{Definition}

\newcommand{\sqlString}{\texttt{string}}
\newcommand{\sqlBool}{\texttt{bool}}
\newcommand{\sqlInt}{\texttt{int}}

\newcommand{\figlabel}[1]{\label{fig:#1}}
\newcommand{\seclabel}[1]{\label{sec:#1}}

\newcommand{\secref}[1]{Sec.~\ref{sec:#1}}  
\newcommand{\figref}[1]{Figure~\ref{fig:#1}}     

\newcommand\sys{\textsc{DopCert}\xspace}
\newcommand\sem{\textsl{HoTT}\textsc{SQL}\xspace}
\newcommand\inputLang{\sem}
\newcommand\outputLang{\textsc{UniNomial}\xspace}

\begin{document}

\setlength{\pdfpageheight}{\paperheight}
\setlength{\pdfpagewidth}{\paperwidth}

\conferenceinfo{CONF 'yy}{Month d--d, 20yy, City, ST, Country}
\copyrightyear{20yy}
\copyrightdata{978-1-nnnn-nnnn-n/yy/mm}
\copyrightdoi{nnnnnnn.nnnnnnn}


\titlebanner{}        
\preprintfooter{}   

\title{\sem: Proving Query Rewrites with Univalent SQL Semantics}

\renewcommand{\thefootnote}{\fnsymbol{footnote}}

\authorinfo{Shumo Chu\footnotemark \and Konstantin Weitz\footnotemark[1] \and Alvin Cheung \and Dan Suciu}
           {University of Washington, USA} 
           {\{chushumo, weitzkon, akcheung, suciu\}@cs.washington.edu \vspace{-1em} }

\maketitle
\footnotetext[1]{These authors contributed equally to this work.}

\renewcommand{\thefootnote}{\arabic{footnote}}
\setcounter{footnote}{0}


\begin{abstract}

Every database system contains a query optimizer that performs query rewrites.
Unfortunately, developing query optimizers remains a highly challenging task.
Part of the challenges comes from the intricacies and rich features of query 
languages, which makes reasoning about rewrite rules difficult.  
In this paper, we propose a machine-checkable denotational semantics for SQL,
the de facto language for relational database,  for 
rigorously validating rewrite rules. 
Unlike previously proposed semantics that are either non-mechanized or only
cover a small amount of SQL language features, our semantics
covers all major features of SQL, including bags, correlated subqueries,
aggregation, and indexes.  Our mechanized semantics, called \sem, 
is based on K-Relations
and homotopy type theory, where we denote relations as
mathematical functions from tuples to univalent types. 
We have implemented \sem in Coq, which takes only
fewer than 300 lines of code 
and have proved a wide range of SQL rewrite rules,
including those from
database research literature (e.g., magic set
rewrites) and real-world query optimizers (e.g., subquery elimination).
Several of these rewrite rules have never been
previously proven correct. 
In addition, while query equivalence is generally undecidable, we have 
implemented an automated decision procedure 
using \sem for conjunctive queries: a
well-studied decidable fragment of SQL that encompasses many 
real-world queries.

\end{abstract}




\section{Introduction}
\seclabel{intro}

From purchasing plane tickets to browsing social networking websites,
we interact with database systems on a daily basis. 
Every database system consists of a query optimizer that
takes in an input query and 
determines the best program, also called a query plan, to execute 
in order to retrieve the 
desired data. Query optimizers typically consists of two components:
a query plan enumerator that generates query plans that are semantically
equivalent to the input query, and a plan selector that chooses 
the optimal plan from the enumerated ones to execute
based on a cost model.

The key idea behind plan enumeration is to apply {\em rewrite rules}
that transform a given query plan into another one, hopefully one with
a lower cost.  While numerous plan rewrite rules have been proposed
and implemented, unfortunately designing such rules remains a highly
challenging task.  For one, rewrite rules need to be {\em semantically
  preserving}, i.e., if a rule transforms query plan $Q$ into $Q'$,
then the results (i.e., the relation) returned from executing $Q$ must
be the same as those returned from $Q'$, and this has to hold for {\em
  all} possible input database schemas and instances.  Obviously,
establishing such proof for any non-trivial query rewrite rule is not
an easy task.

Coupled with that, the rich language constructs and subtle semantics of
SQL, the de facto programming language used to interact with relational 
database systems, 
only makes the task even more difficult.
As a result, while various rewrite rules
have been proposed and studied extensively 
in the data management research community~\cite{MumickFPR90SIGMOD, Muralikrishna92VLDB, LevyMS94VLDB, SeshadriHPLRSSS96SIGMOD},
to the best of our knowledge only the trivial ones have been 
formally proven to be semantically preserving. This has unfortunately 
led to dire consequences as incorrect query results have been returned 
from widely-used database systems due to unsound rewrite rules, and 
such bugs can often go undetected for extended periods of time
\cite{GanskiW87SIGMOD, MySQLBug, PostgresBug}.

In this paper we describe a system to formally verify the equivalence
of two SQL expressions.  We demonstrate its utility by proving correct
a large number of query rewrite rules that have been described in the
literature and are currently used in popular database systems. We also show
that, given counter examples, common mistakes made in query optimization fail to pass our
formal verification, as they should. Our system shares similar 
high-level goals of building formally verified systems using theorem
provers and proof assistants as recent work has demonstrated
\cite{Leroy09JACM, SEL4, FSCQ}.

The biggest challenge in designing a formal verification system for
equivalence of SQL queries is choosing the right SQL semantics.  Among
the various features of SQL, the language uses both set and bag
semantics and switches freely between them, making semantics
definition of the language a difficult task.  Although SQL is an ANSI
standard~\cite{sql2011}, the ``formal'' semantics defined there is of
little use for formal verification: it is loosely described in English
and has resulted in conflicting interpretations \cite{Date89AW}.  
Researchers have defined two quite different rigorous semantics
of SQL.  The first comes from the formal methods 
community~\cite{MalechaMSW10POPL, VeanesGHT09ICFEM, VeanesTH10LAPR16},
where SQL relations are interpreted as lists, and SQL queries as
functions over lists; two queries are equivalent if they return lists
that are equal up to element permutation (under bag semantics) or up
to element permutation and duplicate elimination (under set
semantics).  The problem with this semantics is that even the simplest
equivalences require lengthy proofs in order to reason about
order-independence or duplicate elimination, and these proofs become
huge when applied to rewrites found in real-world optimizations.
The second semantics comes from the database theory community and uses
only set semantics \cite{TheAliceBook, BunemanLSTW94SIGMOD,
  NegriPS91TODS}.  This line of work has led to theoretical results
characterizing languages for which query equivalence is decidable (and
often fully characterizing the complexity of the equivalence problem),
and separating them from richer languages where equivalence is
undecidable~\cite{ChandraM77STOC,SagivY80JACM,DBLP:conf/icdt/Ullman97,DBLP:conf/icdt/GeckKNS16}.
For example, equivalence is decidable (and $\Pi^P_2$-complete for a
fixed database schema~\cite{DBLP:conf/icdt/Ullman97}, and
coNEXPTIME-complete in general~\cite{DBLP:conf/icdt/GeckKNS16}) for
conjunctive queries with safe negation, but undecidable for
conjunctive queries with unsafe negation.  Unfortunately, this
approach is of limited use in practice, because most query
optimization rules use features that places them in the undecidable
language fragments.

This paper contributes a new semantics for SQL that is both simple
and allows simple proofs of query equivalence. We then demonstrates its
effectiveness by proving the correctness of various powerful query
optimization rules described in the literature. 

Our semantics consists of two non-trivial generalizations of
$K$-relations.  $K$-relations were introduced by Green et al. in the
database theory community~\cite{GreenKT07PODS}, and represent a
relation as mathematical function that takes as input a tuple and
returns its multiplicity in the relation, with 0 meaning that the
tuple does not exist in the relation.  A $K$-relation is required to
have finite support, meaning that only a finite set of tuples have
multiplicity $>0$.  $K$-relations greatly simplify reasoning about
SQL: under set semantics, a relation is simply a function that returns
0 or 1 (i.e., a Boolean value), while under bag semantics it returns a
natural number (i.e., a tuple's multiplicity).  
Database operations such as join or union become
arithmetic operations on the corresponding multiplicities: join
becomes multiplication, union becomes addition.  Determining the
equivalence of a rewrite rule that transforms a query $Q$ into another
query $Q'$ reduces to checking the equivalence of the functions they
denote. For example, proving that the join operation is associative
reduces to proving that multiplication is associative.  As we will
show, reasoning about functions over cardinals is much easier than
writing inductive proofs on data structures such as lists.

However, $K$-relations as defined by~\cite{GreenKT07PODS} are
difficult to use in proof assistants, because one needs to prove for
every SQL expression under consideration that the $K$-relation it
returns has finite support: this is easy with pen-and-paper, but very
hard to encode for a proof assistant. Without a guarantee of finite
support, some operations are undefined, for example projection on an
attribute requires infinite summation.  Our first generalization of
$K$-relations is to drop the finite support requirement, and meanwhile
allow the multiplicity of a tuple to be any cardinal number as opposed
to a (finite) natural number.  Then the possibly infinite sum
corresponding to a projection is well defined.  With this change, SQL
queries are interpreted over finite and infinite bags, where some
tuples may have infinite multiplicities.  To the best of our
knowledge, ours is the first SQL semantics that interprets relations
as both finite and infinite; we discuss some implications in
Sec.~\ref{sec:discussion}.

Our  second generalization of $K$-relations is to replace cardinal
numbers with univalent types.  Homotopy Type Theory
(HoTT)~\cite{hottBook} has been introduced recently as a
generalization of classical type theory by adding membership and
equality proofs.  A {\em univalent type} is a cardinal number (finite
of infinite) together with the ability to prove equality.

To summarize, we define a SQL semantics where a relation is
interpreted as a function mapping each tuple to a univalent type,
whose cardinality represents the multiplicity of the tuple in the
relation, and a SQL query is interpreted as a function from input
relations to an output relation.  We call the SQL language with this
particular semantics \sem.  Our language covers all major features of
SQL.  In addition, since univalent types have been integrated into the
Coq proof assistant, we leverage that implementation to prove
equivalences of SQL expressions.

To demonstrate the effectiveness of \sem, we implemented a new system
called \sys (Database OPtimizations CERTified) for proving equivalence
of SQL rewrite rules.  We have used \sys to prove many well-known
and commonly-used rewrite rules from the data management research
literature, many of which have never been formally proven correct
before:
aggregates~\cite{ChaudhuriD97SIGMOD,DBLP:conf/pods/KhamisNR16}, magic
sets rewriting~\cite{BancilhonMSU86PODS}, query rewriting using
indexes~\cite{TsatalosSI94}, and equivalence of conjunctive
queries~\cite{TheAliceBook}.  All our proofs require at most a few
dozens lines of Coq code using \sys, as shown in Fig.~\ref{fig:rules}.  All
definitions and proofs presented in this paper are open-source and 
available online.\footnote{\url{http://dopcert.cs.washington.edu}} 

In summary, this paper makes the following contributions:

\begin{itemize}
\item We present \sem, a (large fragment) of SQL whose semantics
  generalizes $K$-relations to infinite relations and univalent types.
  The goal of this semantics is to enable easy proofs for the
  equivalence of query rewrite rules. (Sec.~\ref{sec:semantics}.)
\item We prove a wide variety of well-known and widely-used SQL
  rewrite rules, where many of them have not be formally proven
  before; each proof require at most a few dozens lines of Coq code
  using \sys. (Sec.~\ref{sec:denotation}.)
\item We implement \sys, a new system written in Coq for checking the
  equivalence of SQL rewrite rules. \sys comes with a number of
  heuristic tactics for deciding the equivalence of arbitrary rewrite
  rules, and a fully automated procedure for deciding rewrite rules
  involving conjunctive queries, where conjunctive queries represent a
  fragment of SQL where equivalence is
  decidable. (Sec.~\ref{sec:rules}.)
\end{itemize}

The rest of this paper is organized as follows. In~\secref{overview},
we given an overview and motivation for a new semantics for SQL.  We
then introduce \sem in~\secref{semantics}, its semantics in
\secref{denotation}, and demonstrate our results in~\secref{rules}.
Related work is presented in \secref{related}.  We include some
discussion in \secref{discussion} and conclude in \secref{conclusion}.


\section{Overview}
\label{sec:overview}

\begin{figure}
\begin{small}
\textbf{Rewrite Rule}:
\begin{flalign*}
\quad\; & \SelectFromWhere{\Star}{(\UNIONALL{R}{S})}{\pred} \quad \equiv & \\
        & \UNIONALL{(\SelectFromWhere{\Star}{R}{\pred})}{(\SelectFromWhere{\Star}{S}{\pred})} &
\end{flalign*}

\textbf{\sem Denotation}:
\begin{flalign*}
\Rightarrow & \lambdaFn{t}{(\denote{R} \; t + \denote{S} \; t) \times \denote{\pred} \; t}
\equiv
\lambdaFn{t}{\denote{R} \; t \times \denote{\pred} \; t + \denote{S} \; t \times \denote{\pred} \; t} &
\end{flalign*}

\textbf{\sem Proof}:
Apply distributivity of $\times$ over $+$.
\end{small}
\caption{Proving a rewrite rule using \sem.
Recall that \texttt{UNION ALL} means bag-union in SQL, which in \sem
is translated to addition of tuple multiplicities in the two input relations.
}
\label{fig:union-slct}
\end{figure}

\paragraph{SQL} The basic datatype in SQL is a {\em relation}, which has a
{\em schema} (a relation name $R$ plus attribute names $\sigma$), and
an {\em instance} (a bag of tuples).  A SQL query maps one or more
input relations to a (nameless) output relation.  For example, if a
relation with schema $R(a,b)$ has instance
$\set{(1,40),(2,40),(2,50)}$ then the SQL query
\begin{flalign*}
  \texttt{Q1:} & \SelectFrom{a}{R}
\end{flalign*}
returns the bag $\set{1,2,2}$.  

SQL freely mixes set and bag semantics, where a set is simply a bag
without duplicates and uses the \texttt{distinct} keyword to remove
duplicates. For example, the query:
\begin{flalign*}
\texttt{Q2:} & \SelectFrom{\texttt{DISTINCT } a}{R}
\end{flalign*}
returns the set $\set{1,2}$.  

\paragraph{List Semantics} Previous approaches to mechanizing formal proofs
of SQL query equivalences represent bags as
list~\cite{MalechaMSW10POPL, VeanesGHT09ICFEM, VeanesTH10LAPR16}.
Every SQL query admits a natural interpretation over lists, using a
recursive definition~\cite{DBLP:journals/tcs/BunemanNTW95}.  To prove
that two queries are equivalent, one uses their inductive definition
on lists, and proves that the two results are equal up to element
reordering and duplicate elimination (for set semantics).

The main challenges in this approach are coming up with the induction
hypothesis, and dealing with list equivalence under permutation and
duplicate elimination.
Inductive proofs quickly grow in complexity, even for simple query
equivalences.  Consider the following query:
\begin{flalign*}
  \texttt{Q3:} & \SelectFromWhere{\texttt{DISTINCT } x.a}{R \AS x, R \AS y}{x.a=y.a}
\end{flalign*}
\texttt{Q3} is equivalent to \texttt{Q2}, because it performs a
redundant self-join: the inductive proof of their equivalence is quite
complex, and has, to the best of our knowledge, not been done formally
before. A much simpler rewrite rule, the commutativity of selection,
requires 65 lines of Coq proof in prior work~\cite{MalechaMSW10POPL}, and only 10
lines of Coq proof in our semantics.
Powerful database query optimizations, such as magic sets rewrites 
and conjunctive query equivalences, are based on generalizations of
redundant self-joins elimination like $\texttt{Q2}\equiv\texttt{Q3}$,
but significantly more complex (see Sec.~\ref{sec:rules}), and
inductive proofs become impractical.
This motivated us to consider a different semantics; we do not use
list semantics in this paper.

\paragraph{$K$-Relation SQL Semantics} An alternative approach introduced
  in~\cite{GreenKT07PODS} is to represent relations as functions that
  map every tuple to a value that indicates how many times it appears
  in the relation.   If the relation is a bag, then the function
  returns a natural number, and if it is a set then it returns a value
  in $\set{0,1}$.  
  More generally, a commutative semi-ring is a structure
  ${\bf K} = (K, +, \times, 0, 1)$ where both $(K,+,0)$ and
  $(K,\times,1)$ are commutative monoids, and $\times$ distributes
  over $+$.  For a fixed set of attributes $\sigma$, denote
  $\Tuple(\sigma)$ the type of tuples with attributes $\sigma$.  A
  $K$-relation~\cite{GreenKT07PODS} is a function:
\begin{align*}
  \denote{R}: & \Tuple \; \sigma  \rightarrow K
\end{align*}
with finite support, meaning that the set
$\setof{t}{\denote{R}\; t \neq 0}$ is finite.  A bag is an
$\N$-relation, and a set is a $\Bool$-relation.  All relational
operators are expressed in terms of the semi-ring operations, for
example:
\begin{align*}
  & \denote{\UNIONALL{R}{S}} =  \lambdaFn{t}{\denote{R} \; t+\denote{S} \; t} \\ 
  & \denote{\SelectFrom{\Star}{R,S}} =  \lambdaFn{(t_1,t_2)}{\denote{R} \; t_1 \times \denote{S} \; t_2} \\
  & \denote{\SelectFromWhere{\Star}{R}{\pred}} = \lambdaFn{t}{\denote{R} \; t \times \denote{\pred} \; t}\\
  & \denote{\SelectFrom{x.a}{R}} =  \lambdaFn{t}{\sum_{t' \in \Tuple \; \sigma}(t = \D{a} \; t') \times \denote{R} \; t'}\\
  & \denote{\SelectFrom{\DISTINCT \Star}{R}} =  \lambdaFn{t}{\merely{\denote{R} \; t}} \\
\end{align*}
where, for any predicate $\pred$: $\denote{\pred} \; t = 1$ if the
predicate holds on $t$, and $\denote{\pred} \; t = 0$ otherwise.  The
function $\merely{\ }$ is defined as $\merely{x}=0$ when $x=0$, and
$\merely{x}= 1$ otherwise (see Subsec.~\ref{sec:target-lang}).  The
projection $\D{a} \; t'$ returns the attribute $a$ of the tuple $t'$,
while equality $(x=y)$ is interpreted as 0 when $x\neq y$ and 1
otherwise.

To prove that two SQL queries are equal one has to prove that two
semi-ring expressions are equal. For example,
Fig.~\ref{fig:union-slct} shows how we can prove that selections
distribute over unions, by reducing it to the distributivity of
$\times$ over $+$, while Fig.~\ref{fig:magic-distinct} shows the proof
of the equivalence for $\texttt{Q2} \equiv \texttt{Q3}$.

\begin{figure}
\begin{small}
\textbf{Rewrite Rule}:
\begin{flalign*}
\quad\;
        & \SelectFromWhere{\texttt{DISTINCT } x.a}{R \AS x, R \AS y}{x.a=y.a}  \quad \equiv &\\
        & \SelectFrom{\texttt{DISTINCT } a}{R}&
\end{flalign*}

\textbf{Equational \sem Proof}:

\begin{flalign*}
\Rightarrow & \lambdaFn{t}{\merely{\sum_{t_1, t_2} (t = \D{a}\; t_1) \times (\D{a} \; t_1= \D{a} \; t_2) \times \denote{R} \; t_1 \times \denote{R} \; t_2}} \equiv& \\
   & \lambdaFn{t}{\merely{\sum_{t_1, t_2} (t = \D{a} \; t_1) \times (t = \D{a} \; t_2) \times \denote{R} \; t_1 \times \denote{R} \; t_2}} \equiv& \\
   & \lambdaFn{t}{\merely{(\sum_{t_1} (t = \D{a} \; t_1) \times \denote{R} \; t_1) \times (\sum_{t_2} (t = \D{a} \; t_2)  \times  \denote{R} \; t_2)}} \equiv& \\
   & \lambdaFn{t}{\merely{\sum_{t_1} (t = \D{a} \;t_1) \times \denote{R} \; t_1}}&
\end{flalign*}

We used the following semi-ring identities:
\begin{flalign*}
(a=b) \times  (b=c) \equiv & (a=b) \times (a=c)\\
\sum_{t_1, t_2} E_1(t_1) \times E_2(t_2) \equiv & \sum_{t_1} E_1(t_1) \times \sum_{t_2} E_2(t_2)\\
\merely{n \times n} \equiv & \merely{n}
\end{flalign*}

\textbf{Deductive \sem Proof}:
\begin{flalign*}
\Rightarrow \forall \; t. & \exists_{t_0} (\denote{a} \; t_0 = t) \land \denote{R} \; t_0 \leftrightarrow & \\
& \exists_{t_1, t_2} (\denote{a} \; t_1 = t) \land \denote{R} \; t_1 \land \denote{R} \; t_2 \land (\denote{a} \; t_1 = \denote{a} \; t_2) &
\end{flalign*}

Then case split on $\leftrightarrow$. 
Case $\rightarrow$: instantiate both $t_1$ and $t_2$ with $t_0$, then apply hypotheses.
Case $\leftarrow$: instantiate $t_0$ with $t_1$, then apply hypotheses.

\end{small}
\caption{The proof of  equivalence  $\texttt{Q2} \equiv \texttt{Q3}$.}
\label{fig:magic-distinct}
\end{figure}

Notice that the definition of projection requires that the relation
has finite support; otherwise, the summation is over an infinite set
and is undefined in $\N$.  This creates a major problem for our
equivalence proofs, since we need to prove, for each intermediate
result of a SQL query, that it returns a relation with finite support.
This adds significant complexity to the otherwise simple proofs of
equivalence.

\paragraph{\sem Semantics} To handle this challenge, our semantics
generalizes $K$-Relation in two ways: we no longer require relations to
have finite support, and we allow the multiplicity of a tuple to be an
arbitrary cardinality (possibly infinite).  More precisely, in our
semantics a relation is interpreted as a function:
\begin{align*}
   & \Tuple \; \sigma \rightarrow \Type
\end{align*}
where $\Type$ is the class of homotopy types.  We call such a
relation a \emph{HoTT-relation}.  A homotopy type $n \in \Type$ is an
ordinary type with the ability to prove membership and equality between
types.  

Homotopy types form a commutative semi-ring and can well represent cardinals.
Cardinal number 0  is the empty homotopy type \textbf{0}, 1 is the unit type \textbf{1}, multiplication is the product type $\times$, addition is the sum type $+$, infinite summation is the dependent product type $\Sigma$, and truncation is the squash type $\merely{.}$. 
Homotopy types generalize  natural numbers and their semiring operations,
and is now well integrated with automated proof assistants like Coq \footnote{After adding the Univalence Axiom to Coq's underlying type theory.}. As we show in the rest of this paper, the equivalence proofs retain the simplicity of $\N$-relations and can be
easily mechanized, but without the need to prove finite support. 

In addition, homotopy type theory unifies squash type and proposition. Using the fact that propositions as types in homotopy type theory \cite[Ch 1.11]{hottBook}, in order to prove the equivalence of two squash types, $\merely{p}$ and $\merely{q}$, it is sufficient to just prove the bi-implication ($p \leftrightarrow q$), which is arguably easier in Coq. For example, transforming the equivalence proof of Figure~\ref{fig:magic-distinct} to bi-implication would not require a series of equational rewriting using semi-ring identities any more, which is complicated because it is under the variable bindings of $\Sigma$. The bi-implication can be proved in Coq by deduction easily.  

The queries of the above rewrite rule fall in the well studied category of conjunctive queries,
for which equality is decidable (equality between arbitrary SQL queries is undecidable).
Using Coq's support for automating deductive reasoning (with \emph{Ltac}), we
have implemented a decision procedure for the equality of conjunctive
queries, the aforementioned rewrite rule can thus be proven in one line of Coq code.


\section{\sem and Its Semantics}
\label{sec:semantics}

In this section, we present \sem, a SQL-like language for expressing rewrite rules.
To simplify discussion, we first describe how relational 
data structures are modeled
in Section~\ref{sec:data-model}. We then describe \sem, a language built 
on top of our relational data structures that covers
all major features of SQL, in Section~\ref{sec:input-lang}. Next, we define
\outputLang, the formal expressions into which \sem is translated, in
Section~\ref{sec:target-lang}.

\subsection{Data Model}
\label{sec:data-model}

We first describe how schemas for relations and tuples are modeled
in \sem. Both of these foundational concepts from relational 
theory~\cite{Codd70CACM} are what \sem uses to build upon.

\begin{figure}[t]
\[
\begin{array}{llll}
  \type \in \texttt{Type}                 & ::=     & \sqlInt \sep \sqlBool \sep \sqlString \sep \; \ldots \\ 
  \denote{\sqlInt}                        & ::=     & \Integer \\
  \denote{\sqlBool}                       & ::=     & \Bool \\
  \denote{\sqlString}                     & ::=     & \String \\
  \ldots & & \\ \\

  \schema \in \Schema                     & ::=     & \emptySchema \\
                                          & \; \mid & \leaf{\type} \\
                                          & \; \mid & \node{\schema_1}{\schema_2} \\ \\

  \Tuple \; \emptySchema                  & ::=     & \Unit \\
  \Tuple \; (\node{\schema_1}{\schema_2}) & ::=     & \Tuple \; \schema_1 \times \Tuple \; \schema_2 \\
  \Tuple \; (\leaf{\type})                & ::=     & \denote{\type} \\
\end{array}
\]
\caption{Data Model of \sem}
\label{fig:data-model}
\end{figure}

\begin{figure}[t]
\begin{minipage}{.4\columnwidth}
\tikzset{
  internode/.style = {shape=circle, draw, align=center}
}
\begin{tikzpicture}
  \begin{scope}
  \node [internode, anchor=north] (root) at (0, 0) {}
    child {node  {\sqlString}} 
    child {node [internode] {} 
      child {node {\sqlInt}} 
      child {node {\sqlBool}} 
    }; 
    
  \end{scope}
\end{tikzpicture}
\end{minipage}
\begin{minipage}{.6\columnwidth}
\begin{small}
\[
\begin{array}{lll}
  \schema : \Schema & = & \node{(\leaf{\sqlInt})}{} \\
                      & & \quad (\node{(\leaf{\sqlInt})}{} \\
                      & & \quad\quad (\leaf{\sqlBool}))  \\
\Tuple \; \schema & = &   \String \times (\Integer \times \Bool) \\
t : \Tuple \; \schema & = &  \mkPair{``Bob''}{\mkPair{52}{\true}}
 \end{array}
\]
\end{small}
\end{minipage}
\caption{An Example of \sem Schema and Tuple}
\label{ex:data-model}
\end{figure}
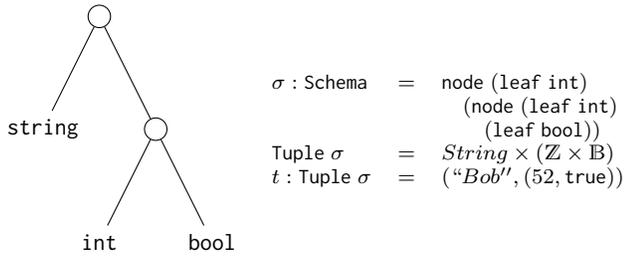

\paragraph{Schema and Tuple}
We briefly review the standard SQL definitions of a schema and a tuple.
Conceptually, a database schema is an unordered bag of $(n, \tau)$ pairs, where
$n$ is an attribute name, and $\tau$ is the type of the attribute.
For example, the schema of a table containing personal information might be: 
\[ \{(\text{Name}, \sqlString), (\text{Age}, \sqlInt), (\text{Married}, \sqlBool)\} \] 

A database tuple is a collection of values that conform to a
given schema. For example, the following is a tuple with the 
aforementioned schema:
\[ \{\text{Name}:\text{``Bob''}; \; \text{Age}:52; \; \text{Married}: \true \} \]

Attributes from tuples are accessed using record syntax. For instance
$t.\text{Name}$ returns ``Bob'' where $t$ refers to the tuple above. 

As shown in Figure~\ref{fig:data-model},
we assume there exists a set of SQL types \texttt{Type}, which can be 
denoted into types in Coq.

In \sem, we define schemas and tuples as follows.  A schema is modeled
as a collection of types
organized in a binary tree, with each type corresponding to an attribute. 
As shown in~\figref{data-model}, a schema can be constructed
from the $\texttt{empty}$ schema,
an individual type $\tau$, or recursively 
from two schema nodes $s_1$ and $s_2$, corresponding to the branches
of the subtree. 
As we will see, this organization is beneficial in 
both writing \sem rewrite rules and also reasoning about the equivalence 
of schemas. 

The tuple type in \sem is defined as a dependent type on a schema. As shown in~\figref{data-model}, a tuple is an nested pair with the identical structure as its schema. Given a schema $s$, 
if $s$ is the empty schema, then the (only) instance of $\Tuple \; \emptySchema$ is  {\tt Unit} (i.e., empty) tuple. 
Otherwise, if $s$ is a leaf node in the schema tree with type $\tau$,
then a tuple is simply a value of the type $\D{\tau}$. 
Finally, if $s$ is recursively
defined using two schemas $s_1$ and $s_2$, then the resulting tuple
is an instance of a product type $\Tuple(s_1) \times \Tuple(s_2)$. 
 
As an illustration,
Figure~\ref{ex:data-model} shows a tuple $\tuple$ and its schema $\schema$, where $\schema$ =
\texttt{(node (leaf \textrm{Int}) (node (leaf \textrm{Int}) (leaf
\textrm{Int})))} and 
$t$ has the nested pair type $(\String \times (\Integer \times \Bool))$.
To access an element from a tuple, \sem uses path
expressions with selectors $\ProjL$ and $\ProjR$. For instance, 
the path $\ProjL.\ProjR$ retrieves
the value 52 from the tuple $t$ in \figref{data-model}.
As will be shown in \secref{denotation}, path expressions will be denoted to standard 
pair operations, i.e., $\ProjL$ will be denoted to ``$\fst{}$'', which returns the first element from 
a pair, and $\ProjL$ will be denoted to ``$\snd{}$'', which returns the second. Such expressions can be
composed to retrieve nested pair types. For instance, $\ProjL.\ProjR$ will be denoted to ``$\snd{\fst{}}$'', thus it retrieves 52 from $t$ in \figref{data-model}.

\paragraph{Relation}
In \sem a relation is modeled as a function from tuples to homotopy
types called HoTT-relations, $\Tuple \; \schema \rightarrow \Type$, as
already discussed in~\secref{overview}; we define homotopy types
shortly.

\paragraph{Discussion}
We briefly comment on our choice of data model.  There are two
approaches to defining tuples in database theory \cite{TheAliceBook}:
the named approach and unnamed approach.  We chose an unnamed
approach, because it avoids name collisions, and because proof
assistants like Coq can decide schema equivalence of unnamed schemas
based on structural equality. Previous work~\cite{MalechaMSW10POPL}
adopted an unnamed approach as well, for the same reason.  Our choice
for representing tuples as trees, rather than as ordered lists, is
non-standard: we do this in order to allow our language to express
generic rewrite rules, without specifying a particular schema for the
input relations, see \secref{rewriteRules}.

Finally, we note that, in our model, a tuple is a dependent type,
which depends  on
its schema.  We use dependent types to ensure that a tuple, which is a nested pair, must have the same structure as its schema, which is a binary tree, by construction. This allows us to denote a path expression (composed by $\ProjL$, $\ProjR$) to a series of corresponding pair operations (composed by ``$\fst{}$'', ``$\snd{}$'') easily.

\subsection{\sem: A SQL-like Query Language}
\label{sec:input-lang}

We now describe \sem, our source language used to 
express rewrite rules. 
Figure~\ref{fig:sql-syntax} defines the syntax of \sem. We divide the
language constructs of \sem 
into four categories: queries, predicates, expressions, and projections.

\paragraph{Queries} A query takes in relation(s) and outputs another relation.
The input to a query can be a base relation (called a $Table$ in~\figref{sql-syntax})
or other queries, including 
projections, cross product, selections, bag-wise operations (\texttt{UNION ALL} and \texttt{EXCEPT}),
and finally conversion to sets (\texttt{DISTINCT}).

\paragraph{Predicates} Predicates are used as part of selections 
(i.e., filtering of tuples) in queries. Given a tuple $t$, predicates return
a Boolean value to indicate whether $t$ should be retained in
the output relation. 

\paragraph{Expressions} Expressions are used both in predicates and projections,
and they evaluate to values (e.g., of type $\Integer$, $\Bool$, $\String$, etc).
Expression includes conversions from projection to expression ($\Var{\proj}$), uninterpreted functions on expressions, 
aggregators of a query, and casts of an expression ($\CastExpr{\proj}{\expr}$). 

$\Var{\proj}$ converts a projection $\proj$ into an expression. For example, $\Var{a} = \Var{b}$ is an equality predicate on attribute $a$ and attribute $b$, where $\Var{a}$ and $\Var{b}$ are the expressions representing attribute $a$ and $b$. As shown next, we use a projection to represent an attribute.

Uninterpreted functions of expressions $f(e_1, \ldots, e_n)$ are used to represent arithmetic operations on expressions such as addition, multiplication, division, mode, and constants (which are nullary uninterpreted functions).

$\texttt{CASTEXPR}$\; is a special construct that is used to represent castings of meta-variables to express generic query rewrite rules. A comprehensive discussion of meta-variables can be found in \secref{rewriteRules}. A normal (non-generic) query would not need $\texttt{CASTEXPR}$.

\paragraph{Projections} When applied to a relation, 
projections denote a subset of attributes to be returned.
A projection is defined to be 
a tuple to tuple function. It can be the identity function (\Star), or 
return the subtree denoted by any of the 
path expressions as discussed in~\secref{data-model}, such as returning the
left (\ProjL) or right (\ProjR) subtree of the tuple.
Any empty tuple can also be produced using \texttt{empty}. Multiple
projections can be composed using ``{\tt .}''.
Two projections $\proj_1$ and $\proj_2$ can be applied to the input tuple separately with the results combined together using ``{\tt ,}''. $\Evaluate{\proj}$ is used to convert a projection to an expression. 
Below are examples of \sem query using projections:

\vspace{2mm}

{\centering
\begin{tabular}{lll}
   & SQL & \inputLang \\ \hline
 $q_1$ & \texttt{SELECT $R.*$ FROM $R, S$} &  \texttt{SELECT Left.* FROM $R, S$} \\ 
 $q_2$ & \texttt{SELECT $S.*$ FROM $R, S$} &  \texttt{SELECT Right.* FROM $R, S$} \\
 $q_3$ & \texttt{SELECT $S.p$ FROM $R, S$} &  \texttt{SELECT Right.$p$ FROM $R, S$} \\
 $q_4$ & \texttt{SELECT $R.p_1$, $S.p_2$} & \texttt{SELECT Left.$p_1$, Right.$p_2$} \\
 &\texttt{FROM $R, S$} & \texttt{FROM $R, S$} \\
 $q_5$ & \texttt{SELECT $p_1+p_2$ FROM $R$} & \texttt{SELECT $\Evaluate add(\Var{p_1},\Var{p_2})$} \\ 
& & \texttt{FROM $R$} \\
\end{tabular}
}

In $q_1$, we compose the path expressions
\texttt{Left} and \texttt{*} to represent projecting all attributes of $R$ from
a tuple that is in the result of $R \join S$ \footnote{Technically 
$\join$ denotes natural join in relational theory~\cite{Codd70CACM}. However, since
\sem schemas are unnamed, there are no shared
names between schemas, and thus natural joins are equivalent to cross products.
Hence, we use $\join$ for cross product of relations to distinguish
it from Cartesian product of types ($\times$) to be discussed in~\secref{target-lang}.}.  In $q_3$,
we compose \texttt{Right} and $k$ to project to a single attribute from $S$. 
The variable $k$ is a projection to a singleton tuple, which is the way to represent attributes in our semantics.
In $q_4$, we are projecting one attribute $\proj_1$ from $R$, and another
attribute $\proj_2$ from $S$ using the projection combinator ``{\tt ,}''.
In $q_5$, to represent $\proj_1+\proj_2$, we first convert attributes ($\proj_1$,
$\proj_2$) to expressions ($\Var{\proj_1}$, $\Var{\proj_2}$), then
use an uninterpreted function $\texttt{add}$ to represent addition, and cast that
back to a projection using $\Evaluate$.

\subsection{Expressing Rewrite Rules}
\seclabel{rewriteRules}

\sem is a language for expressing query rewrite rules, and each such
rule needs to hold over all relations (i.e., both sides of each rule
need to return the same relation for all schemas and instances, as
discussed in~\secref{overview}), and likewise for predicates and
expressions.  To facilitate that, \sem allows users to declare
meta-variables for queries, predicates, and expressions, and uses two
functions $\CastPred{}{}$ and $\CastExpr{}{}$, as we  illustrate next.

First, consider meta-variables.  Referring to~\figref{union-slct}, the
base tables $R$ and $S$ are meta-variables that can be quantified over
all possible relations, and $\pred$ is a meta-variable ranging over
predicates.
For another example, consider the rewrite rule in
Fig.~\ref{fig:magic-distinct}: we want to say that the rule holds for
any relation $R$ with a schema having an attribute $a$.  We express
this in \sem by using a meta-variable $p$ instead of the attribute $a$:
\begin{flalign*}
\quad\;
  & \texttt{DISTINCT} \; \SelectFrom{\ProjL.p}{R, R} \\
  & \texttt{WHERE \Var{\ProjL.p} = \Var{\ProjR.p}}\\
  &  \equiv \   \texttt{DISTINCT} \; \SelectFrom{p}{R}
\end{flalign*}

Since our data model is a tree, it is very convenient to use a
meta-variable $p$ to navigate to any leaf (corresponding to an
attribute), and it also easy to concatenate two schemas using the
$\node{}{}$ constructor; in contrast, a data model based on ordered
lists would make the combination of navigation and concatenation more
difficult.

Second, we explain the functions $\CastPred{}{}$ and $\CastExpr{}{}$
in \sem with another example.  Consider the rewrite rule (to be
presented in \secref{basic-rewrite}) for pushing down selection
predicates, written informally as:
\begin{small}
\begin{flalign*}
\quad\;
        & \SelectFrom{\texttt{*}}{R, \; (\SelectFromWhere{\texttt{*}}{S}{b})} \quad \equiv &\\
        & \SelectFromWhere{\texttt{*}}{R, \; S}{b} &
\end{flalign*}
\end{small}
In this rule, $b$ is a meta-variable that ranges over all
possible Boolean predicates.  In standard SQL, the two occurrences of
$b$ are simply identical expressions, but in \sem the second
occurrence is in a environment that consists of the schemas of both
$R$ and $S$.
In \sem, this is done rigorously using the $\CastPred{}$ construct:
\begin{small}
\begin{flalign*}
\quad\;
        & \SelectFrom{\Star}{R, \; (\SelectFromWhere{\Star}{S}{b})} \quad \equiv &\\
        & \SelectFromWhere{\Star}{R, \; S}{(\CastPred{\ProjR}{b})} &
\end{flalign*}
\end{small}
The expression $\CastPred{\ProjR}{b}$ is function composition: it
applies $\ProjR$ first, to obtain the schema of $R$, then evaluates
the predicate $b$ on the result.

Requiring explicit casts is an important feature of \sem: doing so ensures that 
rewrite rules are only applicable in situations where they are valid. In the example
above, the rule is only valid for all predicates $b$ where $b$ refers only to attributes
in $S$, and the cast operation makes that explicit. The {\tt CASTEXPR} construct works similarly
for expressions.

\begin{figure}[t]
\centering
\[
\begin{array}{llll}
  q \in \texttt{Query} & ::=  &  \texttt{Table} \\
        & \; \mid & \texttt{SELECT} \; p \; q      \\ 
        & \; \mid & \texttt{FROM} \; q_1, \ldots, q_n      \\
        & \; \mid & q \; \texttt{WHERE} \; p        \\
        & \; \mid & q_1 \; \texttt{UNION ALL} \; q_2   \\
        & \; \mid & q_1 \; \texttt{EXCEPT} \; q_2      \\
        & \; \mid & \texttt{DISTINCT} \; q \\
 b \in \texttt{Predicate} & ::= & e_1 \; \texttt{=} \; e_2 \\
       & \; \mid &  \texttt{NOT} \; b 
            \mid b_1 \; \texttt{AND} \; b_2 
            \mid b_1 \; \texttt{OR} \; b_2  
            \mid \texttt{true} \mid \texttt{false} \\
       & \; \mid & \texttt{CASTPRED} \; p \; b \\
       & \; \mid & \texttt{EXISTS} \; q \\
 e \in \texttt{Expression} & ::= & \Var \; p \\
       & \; \mid & f(e_1, \ldots, e_n)   
            \mid agg(q)   \\
       & \; \mid & \texttt{CASTEXPR} \; p \; e \\
 p \in \texttt{Projection} & ::= & \Star \mid \ProjL \mid \ProjR \mid \texttt{Empty} \\
       & \; \mid &  \Compose{p_1}{p_2} \\
       & \; \mid & \Duplicate{p_1}{p_2} \\
       & \; \mid & \Evaluate{e}  
\end{array}
\]
\caption{Syntax of \inputLang}
\label{fig:sql-syntax}
\end{figure}

\subsection{UniNomials}
\label{sec:target-lang}

The interpretation of a \inputLang expression is a formal expression over \outputLang, 
which is an algebra of univalent types.

\begin{definition}
  \outputLang, the algebra of Univalent Types, consists of
  $(\Type, \zero, \one, +, \times, \negate{\cdot}, \merely{\cdot},
  \sum)$, where:
  \begin{itemize}
  \item $(\Type, \zero, \one, +, \times)$ forms a semi-ring, where
    $\Type$ is the universe of univalent types, $\zero, \one$ are the
    empty and singleton types, and $+, \times$ are binary operations
    $\Type \times \Type \rightarrow \Type$: $\nomial_1 + \nomial_2$,
    is the direct sum, and $\nomial_1 \times \nomial_2$ is the
    Cartesian product.
  \item $\negate{\cdot}, \merely{\cdot}$ are derived unary operations
    $\Type \rightarrow \Type$, where $(\negate{\zero}) = \one$ and
    $(\negate{\nomial}) = \zero$ when $\nomial\neq 0$, and
    $\merely{\nomial}=\negate{(\negate{\nomial})}$.
  \item $\sum : (A \rightarrow \Type) \rightarrow \Type$ is the infinitary
    operation: $\sum f$ is the direct sum of the set of types
    $\setof{f(a)}{a \in A}$.
  \end{itemize}
\end{definition}

Following standard notation~\cite{hottBook}, we say that the homotopy type $n$ inhabits 
some universe $\Type$. The base cases of $n$ come from the
denotation of HoTT-Relation and equality of two tuples (In HoTT, propositions are squash types, which are {\zero} or {\one} \cite[Ch 1.11]{hottBook}). The denotation of \sem will be shown Section~\ref{sec:denotation}. There are 5 type-theoretic operations on $\Type$:

\paragraph{Cartesian product ($\times$)} Cartesian product of univalent types is
analogously the same concept as the Cartesian product of two sets. For 
$A, B:\Type$, the cardinality of $A \times B$ is the cardinality of $A$ multiplied by
the cardinality of $B$. For example, we denote the cross product of two
HoTT-Relations using the Cartesian product of univalent types:

\[ \denote{R_1 \bowtie R_2} \teq \introt (\denote{R_1} \; \fst{t}) \times (\denote{R_2} \; \snd{t})  \]

The result of $R_1 \bowtie R_2$ is a HoTT-Relation with type 
$\Tuple \; \schema_{R_1 \bowtie R_2}$ $\rightarrow \Type$. 
For every tuple 
$t \in R_1 \bowtie R_2$, its cardinality equals to the cardinality of $t$'s left
sub-tuple ($\fst{t}$) in $R_1$ ($\denote{R_1} \; \fst{t}$) multiplied by the
cardinality of $t$'s right sub-tuple ($\snd{t}$) in $R_2$ ($\denote{R_2} \;
\snd{t}$).

\paragraph{Disjoint union ($+$)} Disjoint union on univalent types
is analogously the same
concept as union on two disjoint sets. For $A,B: \Type$, the cardinality of $A
+ B$ is the cardinality of $A$ adding the cardinality of $B$. For example, 
\texttt{UNION ALL} denotes to $+$: 

\[ \denote{\UNIONALL{R_1}{R_2}} \teq \introt (\denote{R_1} \; \tuple) + (\denote{R_2} \; t)  \]

In SQL, \texttt{UNION ALL} means bag semantic union of two relations. Thus a
tuple $t \in \UNIONALL{R_1}{R_2}$ has a cardinality of its
cardinality in $R_1$ ($\denote{R_1} \; \tuple$) added by its cardinality in $R_2$
($\denote{R_2} \; \tuple$). 

We also denote logical \texttt{OR} of two predicates using $+$. $A+B$ is corresponded
type-theoretic operation of logical OR if $A$ and $B$ are squash types (recall that 
squash types are {\zero} or \one ~\cite[Ch 1.11]{hottBook}).

\paragraph{Squash ($\merely{n}$)} Squash is a type-theoretic operation that
truncates a univalent type to {\zero} or \one. For $A:\Type$, $\merely{A} =
\zero$ if $A$'s cardinality is zero and $\merely{A} = \one$ otherwise. An
example of using squash types is in denoting \texttt{DISTINCT} (\texttt{DISTINCT}
means removing duplicated tuples in SQL, i.e., converting a bag to a set):

\[ \denote{\DISTINCT{R}} \teq \introt \merely{\denote{R} \; \tuple} \]

For a tuple $t \in \DISTINCT{R}$, its cardinality equals to $1$ if its cardinality in $R$ is non-zero and equals to $0$ otherwise. This is exactly $\merely{\denote{R}\; \tuple}$.

\paragraph{Negation ($\negate{n}$)} If $n$ is a squash type, $\negate{n}$ is the
negation of $n$. We have $\negate{\zero} = \one$ and $\negate{\one} = \zero$.
Negation is used to denote negating a predicate and to denote \texttt{EXCEPT} . For
example, \texttt{EXCEPT} is used to denote negation:

\[ \denote{R_1 \; \texttt{EXCEPT} \; R_2} \teq \introt (\denote{R_1} \; \tuple) \times (\negate{\merely{\denote{R_2} \; \tuple}}) \]

A tuple $t \in R_1 \; \texttt{EXCEPT} \; R_2$ retains its multiplicity 
in $R_1$ if its multiplicity in $R_2$ is
not 0 (since if $\denote{R_2} \; \tuple \neq \zero$, then ${\merely{\denote{R_2}  \; \tuple} \rightarrow  \zero} = \one $). 

\paragraph{Summation ($\sum$)} Given $A:\Type$ and $B:A \rightarrow \Type$,
\;$\sum_{x:A} B(x)$ is a dependent pair type 
$\sum$ is used to denote projection. For example:

\[ \denote{\texttt{SELECT $k$ FROM $R$}} \teq 
  \introt \sum_{\tuple': \texttt{Tuple} \; \sigma_R} 
    \merely{\denote{k} \; t' = t} \times \denote{R} \; t' \]

For a tuple $t$ in the result of this projection query, its cardinality is the summation
of the cardinalities of all tuples of schema $\sigma_A$ that also has the
same value on column $k$ with $\tuple$. Here $\merely{\denote{k} \; \tuple' =
\tuple}$ equals to \one~if $\tuple$ and $\tuple'$ have same value on $k$,
otherwise it equals to \zero. Unlike $K$-Relations,  
using univalent types allow us to support summation over an infinite domain
and evaluate expressions such as the projection described above.

In general, proving rewrite rules in \outputLang enables us to use 
powerful automatic proving techniques such as associative-commutative term 
rewriting in
semi-ring structures (recall that $\Type$ is a semi-ring) similar to the 
\texttt{ring} tactic~\cite{ringtac} and
Nelson-Oppen algorithm on congruence closure~\cite{NelsonO80JACM}. Both 
of which mitigate our proof burden. 


\section{Denotation of \sem}
\label{sec:denotation}

In this section we define the denotational semantics of \sem. We 
first discuss the translation of \sem constructs into \outputLang. 
Then, in~\secref{more-sql}, 
we describe how advanced features of SQL (such as integrity constraints
and indexes) can be expressed using
\sem and subsequently translated.

Figure~\ref{fig:denote-query} shows the translation of \sem to
\outputLang. The translation rules make use of contexts. A {\em context schema} $\Context$ is a schema (see the
definition of $\Schema$ in Fig.~\ref{fig:data-model}); a {\em context}
$g$ is a tuple of type $\Tuple\; \Context$ associated to that schema.
Intuitively, the context consists of the concatenation of all tuple
variables occurring in a surrounding scope.

For example, consider the \sem query with correlated subqueries
in~\figref{contexts} where path expressions are used to refer to
relations in predicates, as discussed in~\secref{data-model}.  As in
standard SQL, evaluation proceeds from the outermost to the innermost
query, starting with query $q_1$. After the \texttt{FROM} clause in
$q_1$ is processed, the context consists of $R_1$'s schema
($\sigma_{R_1}$), which is then passed to the query $q_2$. In turn,
$q_2$ then processes its \texttt{FROM} clause, and appends the schema
of $R_2$ to the context ($\node{\sigma_{R_1}}{\sigma_{R_2}}$), and
this context is used to evaluate the path expression $q_2$'s predicate
(\texttt{right.$k$ = left.$k$}, i.e., \texttt{$R_1$.k = $R_2$.k}), and
similarly when $q_3$ evaluates its predicate.

In our system, contexts are implemented as tuples. 
To make passing of contexts explicit, in the following,
each \sem construct takes in a context tuple (represented by {\tt Tuple $\Context$}), 
and is translated to functions that take in 
both a tuple ($t$) and a context tuple ($g$).

\begin{figure}
\begin{small}
\[
\begin{array}{ll} 
  \texttt{SELECT  $*$  FROM  $R_1$ WHERE}                          & q_1 \\ 
  \quad {\color{blue}\texttt {-- predicate in $q_2$: $R_2$.b = $R_1$.a }} & \\ 
  \quad \texttt {EXISTS SELECT $*$ FROM $R_2$ WHERE right.$p_2$=left.$p_1$ AND} & q_2\\
  
  \qquad {\color{blue}\texttt {-- predicate in $q_3$: $R_3$.c=$R_2$.b }} & \\ 
  \qquad \texttt{EXISTS  SELECT $*$ FROM $R_3$} & \\
  \hspace{0.7in} \texttt{WHERE right.$p_3$=left.right.$p_2$}  & q_3 
\end{array}   
\]
\end{small}
\hrule
\vspace{0.1in}
\centering
\begin{tabular}{|c|l|}\hline
Query &  Context schema \\ \hline
init & \texttt{$\Context_0$=empty} \\ \hline
$q_1$ & \texttt{$\Context_1$=node $\Context_0$ $\sigma_{R_1}$} \\ \hline
$q_2$ & \texttt{$\Context_2$=node $\Context_1$ $\sigma_{R_2}$} \\ \hline
$q_3$ & \texttt{$\Context_3$=node $\Context_2$ $\sigma_{R_3}$} \\ \hline
\end{tabular}
\figlabel{contexts}
\caption{Using Contexts in Evaluating Correlated Subqueries}
\end{figure}

\subsection{Denoting Basic HoTTSQL Constructs}

\begin{figure*}[t]
\centering
\[
\begin{array}{llll}

   \multicolumn{3}{l}{ \framebox[1.1\width]{ 
    $\denoteQuery{\Context}{\query}{\schema} : \Tuple \; \Context \rightarrow \Tuple \; \schema \rightarrow \Type $} }   & \texttt{(* $\Query$ *)} \\ \\

  \denoteQuery{\Context}{table}{\schema} & \teq &
    \intros \denoteTable{table} \; t \\

  \denoteQuery{\Context}{\SELECT{\proj}{\query}}{\schema} & \teq & 
    \multicolumn{2}{l}{
    \intros \sum_{\tuple':\Tuple \; \schema'}{
      (\denoteProj{\proj}{\node{\Context}{\schema'}}{\schema} \; \mkPair{\context}{\tuple'} = \tuple) \times
      \denoteQuery{\Context}{\query}{\schema'}} \; \context \; \tuple'} \\
  
  \denoteQuery{\Context}{\FROM{\query_1, \query_2}}{\node{\schema_1}{\schema_2}} & \teq &
    \intros
      \denoteQuery{\Context}{\query_1}{\schema_1} \; \context \; \fst{\tuple} \times
      \denoteQuery{\Context}{\query_2}{\schema_2} \; \context \; \snd{\tuple} \\

  \denoteQuery{\Context}{\FROM{\query}}{\schema} & \teq & 
    \intros \callQuery{\Context}{\query}{\schema} \\

  \denoteQuery{\Context}{\WHERE{\query}{\pred}}{\schema} & \teq & 
    \intros
      \callQuery{\Context}{\query}{\schema} \times
      \denotePred{\node{\Context}{\schema}}{\pred} \; \mkPair{\context}{\tuple} \\

  \denoteQuery{\Context}{\UNIONALL{\query_1}{\query_2}}{\schema} & \teq & 
    \intros
      \callQuery{\Context}{\query_1}{\schema} +
      \callQuery{\Context}{\query_2}{\schema} \\

  \denoteQuery{\Context}{\EXCEPT{\query_1}{\query_2}}{\schema} & \teq & 
    \intros
      \callQuery{\Context}{\query_1}{\schema} \times
      (\negate{(\callQuery{\Context}{\query_2}{\schema})}) \\

  \denoteQuery{\Context}{\DISTINCT{\query}}{\schema} & \teq & 
    \intros \merely{\callQuery{\Context}{\query}{\schema}} \\ \\

   \multicolumn{3}{l}{ \framebox[1.1\width]{$ \denotePred{\Context}{\pred} 
   :  \Tuple \; \Context \rightarrow \Type $  } } & \texttt{(* $\Pred$ *)}\\
  \\
  \denotePred{\Context}{e_1 = e_2} & \teq & 
    \intro(\callExpr{\Context}{e_1}{\type} = \callExpr{\Context}{e_2}{\type}) \\

  \denotePred{\Context}{\AND{\pred_1}{\pred_2}} & \teq & 
    \intro \callPred{\Context}{\pred_1} \times \callPred{\Context}{\pred_2} \\

  \denotePred{\Context}{\OR{\pred_1}{\pred_2}} & \teq & 
    \intro \merely{\callPred{\Context}{\pred_1} + \callPred{\Context}{\pred_2}} \\

  \denotePred{\Context}{\NOT{\pred}} & \teq & 
    \intro \negate{(\callPred{\Context}{\pred})}   \\

  \denotePred{\Context}{\EXISTS{\query}} & \teq & 
    \intro \merely{\sum_{t:\Tuple \; \schema}{\callQuery{\Context}{\query}{\schema}}} \\

  \denotePred{\Context}{\FALSE} & \teq & \intro \zero \\

  \denotePred{\Context}{\TRUE} & \teq & \intro \one \\
 
  \denotePred{\Context}{\CastPred{\proj}{\pred}}  & \teq  &
    \intro \denotePred{\Context'}{\pred} \; (\denoteProj{\proj}{\Context}{\Context'} \; \context) \\

   \\
  \multicolumn{3}{l}{ \framebox[1.1\width]{$ \denoteExpr{\Context}{\expr}{\type} :
    \Tuple \; \Context \rightarrow \denote{\type} $}  } & \texttt{(* $\Expr$ *)}  \\ \\

  \denoteExpr{\Context}{\Var \; \proj}{\type} & \teq & 
     \intro \denoteProj{\proj}{\Context}{\leaf \; \type} \; \context   \\

  \denoteExpr{\Context}{f(e_1, \ldots)}{\type} & \teq & 
    \intro  \denote{f}(\callExpr{\Context}{e_1}{\type_1}, \ldots )  \\

  \denoteExpr{\Context}{agg(\query)}{\type'} & \teq & 
    \intro \denote{agg} \; ( \denoteQuery{\Context}{\query}{\leaf \; \type} \; g ) \\

  \denoteExpr{\Context}{\CastExpr{\proj}{\expr}}{\type} & \teq &
    \intro \denoteExpr{\Context'}{\expr}{\type} \; (\denoteProj{c}{\Context}{\Context'} \; g) \\ \\

  \multicolumn{3}{l}{ \framebox[1.1\width]{$ 
    \denoteProj{\proj}{\Context}{\Context'} : \Tuple \; \Context \rightarrow \Tuple \; \Context'$ } }  & \texttt{(* $\Proj$ *)}  \\ \\
 
  \denoteProj{*}{\Context}{\Context} & \teq  &  \intro \context  \\

  \denoteProj{\ProjL}{\node{\Context_0}{\Context_1}}{\Context_0} & \teq & \intro \fst{\context} \\

  \denoteProj{\ProjR}{\node{\Context_0}{\Context_1}}{\Context_1} & \teq & \intro \snd{\context} \\

  \denoteProj{\Empty}{\Context}{\emptySchema} & \teq & \intro \mkUnit \\

  \denoteProj{\Compose{\proj_1}{\proj_2}}{\Context}{\Context''} & \teq & \intro
    \denoteProj{\proj_2}{\Context'}{\Context''}\;(\denoteProj{\proj_1}{\Context}{\Context'} \; \context) \\

  \denoteProj{\Duplicate{\proj_1}{\proj_2}}{\Context}{\node{\Context_0}{\Context_1}} & \teq & 
    \intro (\denoteProj{\proj_1}{\Context}{\Context_0} \; \context, \; \denoteProj{\proj_2}{\Context}{\Context_1} \; \context) \\

  \denoteProj{\Evaluate{\expr}}{\Context}{\leaf \; \type} & \teq & \intro \callExpr{\Context}{\expr}{\type} \\

\end{array}
\]
\caption{Denotational Semantics of \inputLang}
\label{fig:denote-query}
\end{figure*}

\paragraph{Queries} A query $\query$ is denoted to a function
from $\query$'s context tuple (of type $\Tuple \; \Context$) to a HoTT-Relation
(of type \( \Tuple \; \schema \rightarrow \Type \)):
 \[ \denoteQuery{\Context}{q}{\schema}: Tuple \; \Context \rightarrow  Tuple \; \schema \rightarrow \Type \]

The \texttt{FROM} clause is recursively denoted to a series of cross products of
HoTT-Relations. Each cross product is denoted using $\times$ as shown in
Section~\ref{sec:target-lang}. For example: 

\[
\begin{array}{ll}
  
  \denoteQuery{\Context}{\FROM{\query_1 , \; \query_2}}{\schema} & \teq  \\
  \intros (\denoteQuery{\Context}{\query_1}{\schema} \; \context \; \fst{\tuple}) \times
          (\denoteQuery{\Context}{\query_2}{\schema} \; \context \; \snd{\tuple}) 
\end{array}
\]

\noindent where $\fst{\tuple}$ and $\snd{\tuple}$ indexes into the context tuple $\Context$
to retrieve the schemas of $q_1$ and $q_2$ respectively.

Note the manipulation of the context tuple in the denotation of
\texttt{WHERE}: for each tuple $t$, we first evaluate $t$ against the
query before \texttt{WHERE}, using the context tuple $\context$. After
that, we evaluate the predicate $b$ by first constructing a new
context tuple as discussed (namely, by concatenating $\Context$ and
$\sigma$, the schema of $q$), passing it the combined tuple
$(\context,t)$. The combination is needed as $t$ has schema $\sigma$
while the predicate $b$ is evaluated under the schema
$\node{\Context}{\sigma}$, and the combination is easily accomplished
as $g$, the context tuple, has schema $\Context$.

\texttt{UNION ALL}, \texttt{EXCEPT}, and \texttt{DISTINCT} are denoted using +,
negation ($\negate{n}$) and merely ($\merely{n}$) on univalent types as shown in
Section~\ref{sec:target-lang}.

\paragraph{Predicates} A predicate $\pred$ is denoted to a function
from a tuple (of type $\Tuple \; \Context$) to a univalent type (of type
$\Type$):
\[ \denotePred{\Context}{\pred} : \Tuple \; \Context \rightarrow \Type  \]

\noindent More specifically, the return type $\Type$
must be a \emph{squash type}~\cite[Ch. 3.3]{hottBook}.
A squash type can only be a type of 1 element,
namely \textbf{1}, and a type of 0 element, namely \textbf{0}. \inputLang
program with the form \texttt{$\query$ WHERE $\pred$} is denoted to the
Cartesian product between a univalent type and a mere proposition. 

As an example, suppose a particular tuple $t$ has multiplicity 3 in
query $q$, i.e., $q\; t = \denote{R}{t} = {\bf 3}$, where ${\bf 3}$ is a univalent 
type. Since predicates are denoted to propositions, applying the tuple to 
the predicate returns either ${\bf 1}$ or ${\bf 0}$, and the overall result 
of the query for tuple $t$ is then either ${\bf 3} \times {\bf 0} = {\bf 0}$,
or ${\bf 3} \times {\bf 1} = {\bf 1}$, i.e., a squash type.

\paragraph{Expressions and Projections} A value expression $e$ is
denoted to a function from a tuple (of type $\Tuple \; \Context$) to its data
type, such as {\tt int} and {\tt bool} \ ($\denote{\type}$):
\[ \denoteExpr{\Context}{\expr}{\type} : \Tuple \; \Context \rightarrow \denote{\type} \]

\noindent A projection $\proj$ from $\Context$ to $\Context'$ is denoted to a function from a tuple of type $\Tuple \; \Context$ to a tuple of type $\Tuple \; \Context'$.
\[\denoteProj{\proj}{\Context}{\Context'} : \Tuple \; \Context \rightarrow \Tuple \; \Context'\]
Projections are recursively defined. 
A projection can be composed by two projections using ``{\tt .}''. The composition of two projection, ``$\Compose{p_1}{p_2}$'', where $p_1$ is a projection from $\Context$ to $\Context'$ and $p_2$ is a projection from $\Context'$ to $\Context''$, is denoted to a function from a tuple of type $\Tuple \; \Context$ to a tuple of type $\Tuple \; \Context''$ as follows:
\[\intro \denoteProj{\proj_2}{\Context'}{\Context''}\;(\denoteProj{\proj_1}{\Context}{\Context'} \; \context) \]
We apply the denotation of $\proj_1$, which is a function of type $\Tuple\;\Context \rightarrow \Tuple \; \Context'$, to the argument of composed projection $\context$, then apply the denotation of $\proj_2$ to the result of application.
A projection can also be combined by two projections using ``{\tt ,}''. The combining of two projection, $\Duplicate{p_1}{p_2}$, is denoted to:
\[ \intro (\denoteProj{\proj_1}{\Context}{\Context_0} \; \context, \; \denoteProj{\proj_2}{\Context}{\Context_1} \; \context) \] 
where we apply the denotation of $\proj_1$ and the denotation of $\proj_2$ to the argument of combined projection ($\context$) separately, and combine their results using the constructor of a pair. 

\subsection{Denoting Derived HoTTSQL Constructs}
\label{sec:more-sql}
\sem supports additional SQL features including group by, integrity constraints, and index.
All such features are commonly utilized in query optimization.

\paragraph{Grouping} Grouping is a widely-used relational operator that projects rows
with a common value into separate groups, and applies an aggregation function
(e.g., average) to each group.
In SQL, this is supported via the \texttt{GROUP BY} operator that takes in the attribute
names to form groups. \sem supports grouping by de-sugaring \texttt{GROUP BY} using
correlated subqueries that returns a single attribute relation, 
and applying aggregation function to the resulting relation~\cite{BunemanLSTW94SIGMOD}. 
Below is an example of such rewrite expressed using SQL:

\[
\begin{array}{ll} 
 \texttt{SELECT $k$, SUM($g$) FROM $R$ GROUP BY $k$}  \vspace{0.05in} \\
 \vspace{0.05in}\hspace{0.8in} \text{rewrites to} \Downarrow \\
 \texttt{SELECT DISTINCT $k$, SUM( SELECT $g$ FROM $R$} \\
  \qquad \qquad \qquad \qquad \qquad \quad \;\; \texttt{ WHERE $R.k = R_1.k$)} \\
 \texttt{FROM $R$ AS $R_1$} 
\end{array}
\] 

\noindent We will illustrate using grouping in rewrite rules in~\secref{agg-rule}.

\paragraph{Integrity Constraints}
Integrity constraints are used in real-world database systems and
facilitate various semantics-based query optimizations~\cite{DarFJST96}. 
\sem supports two
important integrity constraints: keys and functional dependency, through
syntactic rewrite. 

A {\em key constraint} requires an attribute to have unique values among all
tuples in a relation. In \sem, a projection $k$ is a key to the relation $R$ if
the following holds:

\[  
\begin{array}{l}
   \texttt{key ($k$) ($R$)} := \\
    \qquad \denoteQuery{\texttt{empty}}{\texttt{SELECT * FROM $R$}}{\schema} = \\     
    \qquad \lbr \texttt{empty} \vdash \texttt{SELECT Left.* FROM $R$, $R$} \\
   \qquad \texttt{ WHERE (\Var{Right.Left.$k$}) = (\Var{Right.Right.$k$)}}: \schema \rbr    
\end{array}
\]

To see why this definition satisfies the key constraint, note that $k$
is a key in $R$ if and only if $R$ equals to its self-join
on $k$ after converting the result into a set using \texttt{DISTINCT}. 
Intuitively, if $k$ is a key, then self-join of $R$ on $k$ will keep all the tuples of
$R$ with each tuple's multiplicity unchanged. Conversely, if $R$ satisfies the 
self-join criteria, then attribute $k$ holds unique values in $R$ and is hence a key. 

\paragraph{Functional Dependencies}
Keys are used in defining functional dependencies and indexes.
A \emph{functional dependency} constraint from attribute $a$ to $b$ requires that for any
two tuples $t_1$ and $t_2$ in $R$, $(t_1.a = t_2.a) \rightarrow (t_1.b = t_2.b)$ 
In \sem, two projections $a$ and $b$ forms a functional dependency in relation $R$ if  
the following holds:

\[
\begin{array}{l}
\texttt{fd ($a\; b$) ($R$)} := \\
\qquad \texttt{key Left.* } \lbr \texttt{empty} \vdash \texttt{DISTINCT SELECT a, b} \\ 
\qquad \qquad \qquad \qquad \quad \texttt{FROM R}:\node{(\leaf{\tau_a})}{(\leaf{\tau_b})} \rbr
\end{array}
\]

If $a$ and $b$ forms a functional dependency, then $a$ should be a key in 
the relation the results from projecting $a$ and $b$ from $R$ followed by
de-duplication. The converse argument follows similarly.

\paragraph{Index} 
An index on an attribute $a$ is a data structure 
that speeds up the retrieval of tuples with a given value of $a$~\cite[Ch. 8]{dbSysBook}. 

To reason about rewrite rules that use indexes, we follow the idea that index can
be treated as a logical relation rather than physical data structure from
Tsatalos et al~\cite{TsatalosSI94}. Since defining index as an relation
requires a unique identifier of each tuple (analogous to a pointer to each
tuple in the physical implementation of an index in database systems), we
define index as a \sem query that projects on the a key of the relation and the
index attribute. For example, if $k$ is a key of relation $R$, an index $I$ of
$R$ on attribute $a$ can be defined as:
\[
\texttt{index}(a, R) := \texttt{SELECT $k$, $a$ FROM $R$}
\]

\noindent In Section~\ref{sec:index-rule}, we show example rewrite rules that utilize 
indexes that are commonly used in query optimizers.


\section{\sys: A Verified System for Proving Rewrite Rules}

\label{sec:rules}

To demonstrate the effectiveness of \sem, we implement \sys, a system written in Coq for checking the equivalence of SQL rewrite rules.  \sys consists of four parts, 1) the denotational semantics of \sem, 2) a library consisting of lemmas and tactics that can be used as building blocks for constructing proofs of arbitrary rewrite rules, 3) a fully automated decision procedure for the equivalence of rewrite rules consisting only of conjunctive queries, and 4) a number of proofs of existing rewrite rules from the database literature and real world systems. 
 
\sys relies on the Homotopy Type Theory Coq library~\cite{HoTTCoq}. Its trusted code base contains 296 lines of specification of \sem. Its verified part contains 405 lines of library code (including the decision procedure for conjunctive queries), and 1094 lines of code that prove well known SQL rewrite rules.  

In the following sections, we first show various rewrite rules and the lemmas they use from the \sys library, and then explain our automated decision procedure.

\subsection{Proving Rewrite Rules in \sys by Examples}

We proved 23 rewrite rules from both the database literature and real world 
optimizers using \sem. Figure~\ref{fig:rules} shows the number of 
rewrite rules that we proved in each category and the average lines of code 
(LOC) required per proof. 
  
 \begin{figure} 
  \centering
  \begin{tabular}{llll}
   Category & No. of rules  & Avg. LOC (proof only) \\ \hline
    Basic             & 8  & 11.1              \\
    Aggregation       & 1  & 50              \\ 
    Subquery          & 2  & 17            \\
    Magic Set         & 7  & 30.3          \\ 
    Index             & 3  & 64              \\
    Conjunctive Query & 2  & 1 (automatic)  \\
    \textbf{Total}    & 23 &   25.2             \\
  \end{tabular}
  \caption{Rewrite rules proved}
  \label{fig:rules}
\end{figure}

The following sections show a sampling of interesting rewrite rules in these categories.
Sec~\ref{sec:basic-rewrite} shows how two basic rewrite rules are proved. 
Sec~\ref{sec:agg-rule} shows how to prove a rewrite rule involving aggregation. 
Sec~\ref{sec:magic} shows how to prove the magic set rewrite rules.
Sec~\ref{sec:index-rule} shows how to state a rewrite rule involving indexes. 
 
\subsubsection{Basic Rewrite Rules}
\label{sec:basic-rewrite}

Basic rewrites are simple rewrite rules that are fundamental building blocks 
of the rewriting system. These rewrites are also very effective in terms of
reducing query execution time. We demonstrate how to prove the correctness of
basic rewrite rules in \sys using two examples: selection push down and
commutativity of joins.

\paragraph{Selection Push Down} Selection push down moves a selection (filter)
directly after the scan of the input table to dramatically reduce the amount of
data in the execution pipeline as early as possible. It is known as one of most
powerful rules in database optimizers~\cite{dbSysBook}.
We formulate selection push down as the following rewrite rule in \inputLang:

\begin{small}
\[
\begin{array}{ll}
\denoteQuery{\Context}{\texttt{SELECT * FROM  $R$ WHERE $p_1$ AND $p_2$ }}{\schema} & \equiv \\
\lbr \gvd \texttt{SELECT * FROM} & \\
  \multicolumn{2}{l}{  \qquad \texttt{(SELECT  * FROM $R$ WHERE $p_1$) WHERE $p_2$ : $\schema$}  \rbr}  \\
\end{array}
\] 
\end{small}

This will be denoted to:

\begin{small}
\[
\begin{array}{ll}
\intros \D{p_1} \mkPair{\context}{\tuple} \times \lbr p_2 \rbr \; \mkPair{\context}{\tuple} \times \D{R} \; g \; t  & \equiv \\
\intros \D{p_2} \; \mkPair{\context}{\tuple} \times ( \D{p_1} \; \mkPair{\context}{\tuple} \times \D{R} \; g \; t)  \\ 
\end{array}
\]
\end{small}

The proof proceeds by functional extensionality \footnote{Function
extensionality is implied by the Univalence Axiom.} and the associativity and
commutativity of $\times$.

\paragraph{Commutativity of Joins} Commutativity of joins allows an optimizer to rearrange 
the order of joins in order to get the join order with best performance. This is one
of the most fundamental rewrite rules that almost every optimizer uses. We formulate the
commutativity of joins in \inputLang as follows:

\begin{small}
\[
\begin{array}{ll} 
\denoteQuery{\Context}{\texttt{SELECT * FROM $R$, $S$}}{\Pair{\schema_R}{\schema_S}} \quad \equiv \\
\denoteQuery{\Context}{\texttt{SELECT Right.Right.*, Right.Left.* FROM $S$, $R$}}{\Pair{\schema_R}{\schema_S}}
\end{array}
\]
\end{small}

Note that the select clause flips the tuples from $S$ and $R$, such that the order of the tuples matches the original query.
This will be denoted to:

\begin{small}
\[
\begin{array}{ll}
\intros \D{R} \; \context \; \fst{\tuple} \times \D{S} \; \context \; \snd{\tuple}  \quad \equiv \\
\intros \sum_{t_1} \D{S} \; \context \; \fst{t_1} \times \D{R} \; \context \; \snd{t_1} \times ((\snd{t_1},\fst{t_1}) = t) 
\end{array}
\]
\end{small}

The proof uses Lemma~\ref{lem:sum_pair} provided by the \sys library.

\begin{lemma}
\label{lem:sum_pair}
Let $A$ inhabit $\Type$, and have $P: A \rightarrow \Type$ be a type family, then we have:
\[
  \sum_{x:A \times B} P \; x = \sum_{x: B \times A} P \; (\snd{x}, \fst{x})
\]
\end{lemma}

Together with the fact that $t_1 = (\fst{t_1}, \snd{t_1})$, the rewrite rule's right hand side becomes:

\begin{small}
\[  \intros \sum_{t'} \D{S} \; \context \; \snd{t'} \times \D{R} \; \context \; \fst{t'} \times (t' = t ) \]
\end{small}

The proof then uses Lemma~\ref{lem:pair_eq} provided by the \sys library.

\begin{lemma}
\label{lem:pair_eq}
Let $A$ and $B$ inhabit $\Type$, and have $P: A \times B \rightarrow \Type$, then we have:
\[
 P\; x = \sum_{x'} P \; x' \times (x' = x)
\]
\end{lemma}

After applying Lemma~\ref{lem:pair_eq}, the right hand side becomes the following,
and we can finish the proof by applying commutativity of $\times$: 
\begin{small}
\[
\intros \D{S} \; \context \; \snd{\tuple} \times \D{R} \; \context \; \fst{\tuple}
\]
\end{small}

\subsubsection{Aggregation and Group By Rewrite Rules} 
\label{sec:agg-rule}

Aggregation and Group By are widely used in analytic queries \cite{ChaudhuriD97SIGMOD}. The standard data analytic benchmark TPC-H \cite{tpch} has 16 queries with group by and 21 queries with aggregation out of a total of 22 queries.  Following is an example rewrite rule for aggregate queries. 
The query on the left-hand side groups the relation $R$ by the column $k$, sums
all values in the $b$ column for each resulting partition, and then removes all
results except the partition whose column $k$ is equal to the constant $l$.
This can be rewritten to the faster query that first removes all tuples from
$R$ whose column $k \neq l$, and then computes the sum.

\begin{small}
\[
\begin{array}{ll}
\lbr \gvd \texttt{SELECT * FROM (SELECT $k$, SUM($b$) FROM $R$ GROUP BY $k$)}  & \\ 
\qquad \; \texttt{WHERE $(\Var{k}) = l$} : \schema \rbr \quad \equiv \\ 
\denoteQuery{\Context}{\texttt{SELECT $k$, SUM($b$) FROM $R$ WHERE $(\Var{k})=l$ GROUP BY $k$ }}{\schema}
\end{array}
\]
\end{small}

As shown in Sec.~\ref{sec:more-sql}, we use a correlated subquery and a unary
aggregate function (which takes a \inputLang query as its input) to represent
aggregation on group by SQL queries. After de-sugaring, the group by query 
becomes \texttt{SELECT DISTINCT $\ldots$}. The rule will thus be denoted to:

\begin{small}
\[
\begin{array}{ll}
\intros (\fst{\tuple} = \denote{l}) \times  \| \sum_{t_1} \denote{R} \; t_1 \times 
 (\fst{\tuple} = \denote{k} \; t_1) \times \\
\qquad \quad (\snd{\tuple} = \denote{\texttt{SUM}} \; (\lambda \; t'. \; \sum_{t_2} (\denote{k} \; t_1 = \denote{k} \; t_2)  \times \\
\qquad \qquad \quad \denote{R} \; t_2 \times (\denote{b} \; t_1  = t')) \| \\
\quad \equiv \quad \\
\intros \| \sum_{t_1} (\denote{k} \; t_1 = \denote{l}) \times \denote{R} \; t_1 \times 
 (\fst{\tuple} = \denote{k} \; t_1) \times \\
\qquad \quad (\snd{\tuple} = \denote{\texttt{SUM}} \; (\lambda \; t'. \; \sum_{t_2} 
(\denote{k} \; t_1 = \denote{k} \; t_2) \times (\denote{k} \; t_2 = \denote{l})
\times \\
\qquad \qquad \quad \denote{R} \; t_2 \times  (\denote{b} \; t_1 = t')) \|
\end{array}
\]
\end{small}

The proof proceeds by functional extensionality, after which both sides become squash types. 
The proof then uses the fundamental lemma about squash types, where 
for all squash types $A$ and $B$,
$(A \leftrightarrow B) \Rightarrow (A = B)$.
It thus suffices to prove by cases the bi-implication ($\leftrightarrow$) of both sides. 
In both cases, instantiate $t_1$ with $t_1$ ($t_1$ is the witness of the $\Sigma$ hypothesis).
It follows that $t.1 = \denote{l} = \denote{k} \; t_1$, and thus that $\denote{k} \; t_2 = \denote{l}$ inside $\texttt{SUM}$.

\subsubsection{Magic Set Rewrite Rules}
\label{sec:magic}  

Magic set rewrites are well known rewrite rules that were originally used in
the recursive query processing in deductive databases \cite{BancilhonMSU86PODS,
RohmerLK86NGC}. It was then used for rewriting complex decision support queries
and has been implemented in commercial systems such as IBM's DB2 database
\cite{SeshadriHPLRSSS96SIGMOD, MumickFPR90SIGMOD}. Below is an example of a
complex magic set rewrite from \cite{SeshadriHPLRSSS96SIGMOD}. 

\noindent Original Query:

\begin{small}
\[
\begin{array}{ll}
\texttt{CREATE VIEW $DepAvgSal$ AS} \\
\qquad \texttt{(SELECT $E.did$, AVG($E.sal$) AS $avgsal$ FROM $Emp$ $E$} \\
\qquad \texttt{GROUP BY $E.did$);}  \\
\texttt{SELECT $E.eid$, $E.sal$ FROM $Emp$ $E$, $Dept$ $D$, $DepAvgSal$ $V$ } \\
\texttt{WHERE $E.did = D.did$ AND $E.did = V.did$ AND $E.age < 30$} \\
\qquad \texttt{AND $D.budget > 100000$ AND $E.sal > V.avgsal$}\\
\end{array}
\]
\end{small}

\noindent Rewritten Query:

\begin{small}
\[
\begin{array}{ll}
\texttt{CREATE VIEW $PartialResult$ AS} \\
\qquad \texttt{(SELECT $E.eid$, $E.sal$, $E.did$ FROM $Emp$ $E$, $Dept$ $D$}\\
\qquad \; \texttt{WHERE $E.did = D.did$ AND $E.age < 30$ AND} \\
\qquad \qquad \texttt{$D.budget > 100000$);}\\
\texttt{CREATE VIEW $Filter$ AS} \\
\qquad \texttt{(SELECT DISTINCT $P.did$ FROM $PartialResult$ $P$);}\\
\texttt{CREATE VIEW $LimitedDepAvgSal$ AS} \\
\qquad \texttt{(SELECT $F.did$, AVG($E.sal$) AS $avgsal$} \\
\qquad \; \texttt{FROM $Filter$ $F$, $Emp$ $E$} \\
\qquad \; \texttt{WHERE $E.did = F.did$} \\
\qquad \; \texttt{GROUP BY $F.did$);}    \\
\texttt{SELECT $P.eid$, $P.sal$} \\ 
\texttt{FROM $PartialResult$ $P$, $LimitedDepAvgSal$ $V$} \\
\texttt{WHERE $P.did = V.did$ AND $P.sal > V.avgsal$}\\
\end{array}
\]
\end{small}

This query aims to find each young employee in a big department ($D.budget > 100000$) whose salary is higher then the average salary in her department. Magic set rewrites use the fact that only the average salary of departments that are big and have young employees need to be computed. As described in \cite{SeshadriHPLRSSS96SIGMOD}, all magic set rewrites can be composed from just three basic rewrite rules on semijoins, namely introduction of $\theta$-semijoin, pushing $\theta$-semijoin through join, and pushing $\theta$-semijoin through aggregation.

Following, we show how to state all three rewrite rules using \sys, and show how to prove two.
We firstly define $\theta$-semijoin as a syntactic rewrite in \sem:

\begin{small}
\[
\begin{array}{ll}
\texttt{$A$ SEMIJOIN $B$ ON $\theta$} \teq \\ 
\texttt{SELECT * FROM $A$ WHERE EXISTS (SELECT * FROM $B$ WHERE $\theta$)}
\end{array}
\]
\end{small}

\paragraph{Introduction of $\theta$-semijoin}
This rules shows how to introduce semijoin from join and selection. 
Using semijoin algebra notation, this rewrite can be expressed as follows:

\[
R_1 \join_{\theta} R_2 \equiv R_1 \join_{\theta} (R_2 \semijoin_{\theta} R_1)
\]

\noindent Using \inputLang, the rewrite can be expressed as follows:

\begin{small}
\[
\begin{array}{ll}
\denoteQuery{\Context}{\texttt{SELECT * FROM $R_2$, $R_1$ WHERE $\theta$}}{\schema} & \quad \equiv \quad \\
\multicolumn{2}{l}{\denoteQuery{\Context}{\texttt{SELECT * FROM ($R_2$ SEMIJOIN $R_1$ ON $\theta$), $R_1$ WHERE $\theta$ }}{\schema}}
\end{array}
\]
\end{small}

\noindent which is denoted to:

\begin{small}
\[
\begin{array}{ll}
 \intros \denote{\theta} \; \mkPair{\context}{\tuple} \times \denote{R_2} \; \context \; \fst{\tuple} \times \denote{R_1} \; \context \; \snd{\tuple} & \quad \equiv \quad \\
\intros \denote{\theta} \; \mkPair{\context}{\tuple} \times \denote{R_2} \; \context \; \fst{\tuple} \times \denote{R_1} \; \context \; \snd{\tuple} \times{} \\
\qquad \quad \merely{\sum_{t_1}\denote{\theta} \; \mkPair{\context}{\mkPair{\fst{\tuple}}{t_1}} \times \denote{R_1} \; \context \; t_1}
\end{array}
\]
\end{small}

\noindent The proof uses Lemma~\ref{lem:hprop_prod} provided by the \sys library.

\begin{lemma}
\label{lem:hprop_prod}
$\forall P, T: \Type$, where $P$ is either $\zero$ or $\one$,  we have:
\[ (T \rightarrow P) \Rightarrow ((T \times P) = T) \]
\end{lemma}

\begin{proof}
Intuitively, this can be proven by cases on $T$.
If $T$ is inhabited, then $P$ holds by assumption, and $T \times \one = T$.
If $T = 0$, then $\zero \times P = \zero$.
\end{proof}

Using this lemma, it remains to be shown that 
$\denote{\theta} \; \mkPair{\context}{\tuple}$ and $\denote{R_2} \; \context \; \fst{\tuple}$ and $\denote{R_1} \; \context \; \snd{\tuple}$ 
imply 
$\merely{\sum_{t_1}\denote{\theta} \; \mkPair{\context}{\mkPair{\fst{\tuple}}{t_1}} \times \denote{R_1} \; \context \; t_1}$.
We show this by instantiating $t_1$ with $t.2$, and then by hypotheses.

\paragraph{Pushing $\theta$-semijoin through join}
The second rule in magic set rewrites is 
the rule for pushing $\theta$-semijoin through join, represented in semijoin algebra as:

\[
(R_1 \join_{\theta_1} R_2) \semijoin_{\theta_2} R_3 \equiv (R_1 \join_{\theta_1}R_2') \semijoin_{\theta_2} R_3 
\]

\noindent where $R_2' = E_2 \semijoin_{\theta_1 \land \theta_2} (R_1 \join R_3)$.
This rule can be written in \inputLang as below:


\begin{small}
\[
\begin{array}{ll}
\intros \lbr \gvd \texttt{(SELECT * FROM $R_1$, $R_2$ WHERE $\theta_1$)} \\ 
\qquad \qquad \quad \texttt{SEMIJOIN $R_3$ ON $\theta_2$}:\Pair{\schema_1}{\schema_2} \rbr  \quad \equiv \quad \\
\intros \lbr \gvd \texttt{(SELECT *} \\
\qquad \qquad \quad \texttt{ FROM $R_1$, ($R_2$ SEMIJOIN (FROM $R_1$, $R_3$) ON $\theta_1$ AND $\theta_2$)} \\
\qquad \qquad \quad \texttt{ WHERE $\theta_1$) SEMIJOIN $R_3$ ON $\theta_2$} :\Pair{\schema_1}{\schema_2} \rbr
\end{array} 
\]
\end{small}

\noindent The rule is denoted to:

\begin{small}
\[
\begin{array}{ll}
\intros \merely{\sum_{t_1} \denote{\theta_2} \; \mkPair{\context}{\mkPair{t}{t_1}} \times \denote{R_3} \; \context \; t_1} \times \\
\qquad \quad \denote{\theta_1} \; \mkPair{\context}{\tuple} \times \denote{R_1} \; \context \; \fst{\tuple} \times \denote{R_2} \; \context \; \snd{\tuple} 
 \qquad \equiv \\
\intros \merely{\sum_{t_1} \denote{\theta_2} \; \mkPair{\context}{\mkPair{t}{t_1}} \times \denote{R_3} \; \context \; t_1} \times \\
\qquad \quad \denote{\theta_1} \; \mkPair{\context}{\tuple} \times \denote{R_1} \; \context \; \fst{\tuple} \times \denote{R_2} \; \context \; \snd{\tuple} \times{} \\
\qquad \quad \| \sum_{t_1} \denote{\theta_1} \; \mkPair{\context}{\mkPair{\fst{t_1}}{\snd{\tuple}}} \times 

\denote{\theta_2} \; \mkPair{\context}{\mkPair{\mkPair{\fst{t_1}}{\snd{\tuple}}}{\snd{t_1}}} 
\\
\qquad \quad \times \denote{R_1} \; \context \; \fst{t_1} \times \denote{R_3} \; \context \; \snd{t_1} \| 
\end{array}
\]

\end{small}

We can prove this rule by using a similar approach to the one used to prove introduction of $\theta$-semijoin:
rewriting the right hand side using Lemma~\ref{lem:hprop_prod}. 
and then instantiating $t_1$ with $(t.1, t_1)$ ($t_1$ is the witness of the $\Sigma$ hypothesis).

\paragraph{Pushing $\theta$-semijoin through aggregation} 
The final rule is pushing $\theta$-semijoin through aggregation:

\newcommand{\gbf}{_{\bar{g}}\mathcal{F}_{\bar{f}}}

\[
\gbf(R_1) \semijoin_{c_1 = c_2} R_2 \equiv \gbf(R_1 \semijoin_{c_1 = c_2} R_2)
\]

\noindent where $\gbf$ is a grouping/aggregation operator ($\gbf$ was firstly defined in \cite{SeshadriHPLRSSS96SIGMOD}), and $\bar{g}$ denotes the group by attributes and $\bar{f}$ denotes the aggregation function. In this rule, one extra condition is that $c_1$ is from the attributes in $\bar{g}$ and $c_2$ is from the attributes of $R_2$. This rule can be written in \inputLang as below:

\begin{small}
\[
\begin{array}{ll}
\lbr \gvd \texttt{(SELECT $c_1$, COUNT($a$) FROM $R_1$ GROUP BY $c_1$)} \\ 
\qquad \quad  \texttt{SEMIJOIN $R_2$ ON $c_1 = c_2$ } :  \schema \rbr \quad \equiv \\
\lbr \gvd \texttt{SELECT $c_1$, COUNT($a$)} \\
\qquad \quad \texttt{FROM ($R_1$ SEMIJOIN $R_2$ ON $c_1 = c_2$) GROUP BY $c_1$} : \schema \rbr
\end{array}
\]
\end{small}

\noindent We omit the proof here for brevity.

\subsubsection{Index Rewrite Rules}
\label{sec:index-rule}

As introduced in Section~\ref{sec:more-sql}, we define an index as a \sem query that projects on the indexed attribute and the primary key of a relation. Assuming $k$ is the primary key of relation $R$, and $I$ is an index on column $a$:
\[ I := \texttt{SELECT $k$, $a$ FROM $R$} \]

\noindent We prove the following common
rewrite rule that converts a full table scan to a lookup on an index and a join:

\begin{small}
\[
\begin{array}{ll}
\denoteQuery{\Context}{\texttt{SELECT * FROM $R$ WHERE $a = l$ }}{\schema} \quad \equiv \\
\lbr \gvd \texttt{SELECT * FROM $I$, $R$} \\
\qquad \quad \texttt{WHERE $a = l$ AND Right.Left.$k$ = Right.Right.$k$ } : \schema \rbr 
\end{array}
\]
\end{small}

\noindent We omit the proof here for brevity.

\subsection{Automated Decision Procedure for Conjunctive Queries}
\label{sec:automation}

\begin{figure*}
\centering
\begin{tabular}{|l|l|l|l|l|} \hline
 & Containment (Set)  & Containment (Bag) & Equivalence (Set) &  Equivalence (Bag) \\ \hline 
Conjunctive Queries & NP-Complete \cite{ChandraM77STOC} & Open & NP-Complete \cite{ChandraM77STOC}  &  Graph Isomorphism \cite{ChaudhuriV92PODS}  \\ \hline
Union of Conjunctive Queries & NP-Complete \cite{SagivY80JACM}  & Undecidable \cite{IoannidisR95TODS}  & NP-Complete \cite{SagivY80JACM}  & Open \\ \hline 
Conjunctive Query with $\neq$, $\geq$, and $\leq$ & 
$\Pi_2^p$-Complete \cite{Meyden92PODS} & Undecidable  \cite{JayramKV06PODS} &  $\Pi_2^p$-Complete  \cite{Meyden92PODS}  & Undecidable \cite{JayramKV06PODS}  \\ \hline
First Order (SQL) Queries & Undecidable \cite{Trakhtenbrot50DANUSSR} & Undecidable & Undecidable & Undecidable \\ \hline
\end{tabular}
\caption{Complexities of Query Containment and Equivalence}
\label{fig:complexity}
\end{figure*} 

\begin{figure}
\tikzset{
  internode/.style = {shape=circle, draw, align=center}
}
\begin{tikzpicture} 
  \begin{scope}
  \node [internode] (root) {} 
    child {node (n1) {$\schema_{R_1}$}} 
    child {node (n2) {$\schema_{R_2}$}}; 
   \node [text width=6cm, anchor=east] at (6.5, 0) {Schema of $t_1$ (left)};
  \end{scope}
  \begin{scope}[xshift=4cm]
   \node [internode] {} 
    child {node [internode] {} 
      child {node (n3) {$\schema_{R_1}$}} 
      child {node (n4) {$\schema_{R_1}$}} 
    } 
    child {node (n5) {$\schema_{R_2}$}}; 
   \node [text width=6cm, anchor=east] at (6.5, 0) {Schema of $t_1$ (right)};
  \end{scope}
 \draw[->, dashed, blue] (n3)--(n1);
 \draw[->, dashed, blue] (n4)--(n1);
 \draw[->, dashed, blue] (n5)--(n2);
\end{tikzpicture}

\begin{tikzpicture} 
  \begin{scope}
  \node [internode] (root) {} 
    child {node (n1) {$\schema_{R_1}$}} 
    child {node (n2) {$\schema_{R_2}$}}; 
   \node [text width=6cm, anchor=east] at (6.5, 0) {Schema of $t_1$ (left)};
  \end{scope}
  \begin{scope}[xshift=4cm]
   \node [internode] {} 
    child {node [internode] {} 
      child {node (n3) {$\schema_{R_1}$}} 
      child {node (n4) {$\schema_{R_1}$}} 
    } 
    child {node (n5) {$\schema_{R_2}$}}; 
   \node [text width=6cm, anchor=east] at (6.5, 0) {Schema of $t_1$ (right)};
  \end{scope}
 \draw[->, dashed, red] (n1)--(n3);
 \draw[->, dashed, red] (n2)--(n5);
\end{tikzpicture}
\caption{The mappings found to prove the conjunctive query example, blue lines show the mapping found to prove left $\rightarrow$ right, red lines show the mapping found to prove right $\rightarrow$ left. }
\label{fig:mappings}
\end{figure}
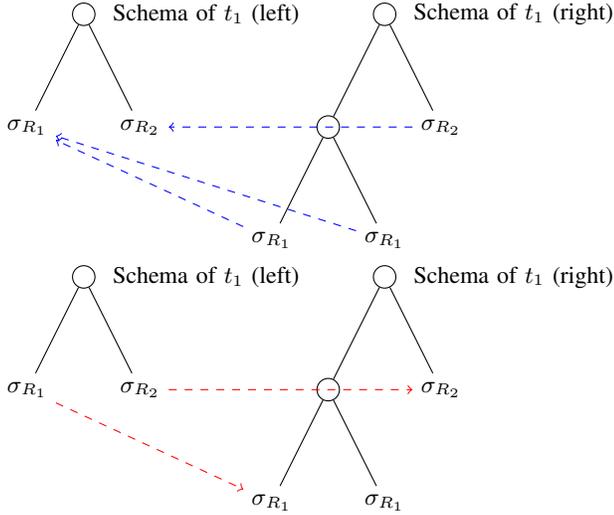

The equivalence of two SQL queries is in general undecidable. Figure~\ref{fig:complexity} shows the complexities of deciding containment and equivalence of subclasses of SQL. The most well-known subclass are conjunctive queries, which are of the form \texttt{DISTINCT SELECT $p$ FROM $q$ WHERE $b$}, where $p$ is a sequence of arbitrarily many attribute projections, $q$ is the cross product of arbitrarily many input relations, and $b$ is a conjunct consisting of arbitrarily many equality predicates between attribute projections. 

We implement a decision procedure to automatically prove the equivalence of conjunctive queries in \inputLang. After denoting the \inputLang query to \outputLang,  
the decision procedure automates steps similar to the proof in Section \ref{sec:agg-rule}. 
First, after applying functional extensionality, both sides become squash types due to the \texttt{DISTINCT} clause. The procedure then applies the fundamental lemma about squash types $\forall A B, (A \leftrightarrow B) \Rightarrow (A = B)$. In both cases of the resulting
bi-implication, the procedure tries all possible instantiations of the $\Sigma$, which is due to the \texttt{SELECT} clause. This search for the correct instantiation is implemented using Ltac's built-in backtracking support. The procedure then rewrites all equalities and tries to discharge the proof by direct application of hypotheses. 

The following is an example of two equivalent conjunctive SQL queries that we can solve using our decision procedure:
\begin{small}
\[
\begin{array}{ll}
\texttt{SELECT DISTINCT $x$.$c_1$ FROM $R_1$ AS $x$, $R_2$ AS $y$} \\
\texttt{WHERE $x$.$c_2$ = $y$.$c_3$} & \quad \equiv \\
\texttt{SELECT DISTINCT $x$.$c_1$ FROM $R_1$ AS $x$, $R_1$ AS $y$, $R_2$ AS $z$} \\
\texttt{WHERE $x$.$c_1$ = $y$.$c_1$ AND $x$.$c_2$ = $z$.$c_3$}
\end{array}
\]
\end{small}
The same queries can be expressed in \inputLang as follows, where the schema of $R_i$ is $\sigma_{R_i}$:
\begin{small}
\[
\begin{array}{ll}
\lbr \gvd \texttt{DISTINCT SELECT Right.Left.$c_1$ FROM $R_1$, $R_2$} \\
\qquad \quad \texttt{WHERE Right.Left.$c_2$ = Right.Right.$c_3$} :\schema \rbr \quad  \equiv  \\
\lbr \gvd \texttt{DISTINCT SELECT Right.Left.Left.$c_1$} \\ 
\qquad \quad \texttt{FROM (FROM $R_1$, $R_1$), $R_2$} \\
\qquad \quad \texttt{WHERE Right.Left.Left.$c_1$ = Right.Left.Right.$c_1$ AND }  \\
\qquad \qquad \qquad \texttt{Right.Left.Left.$c_2$ = Right.Right.$c_3$} : \schema \rbr
\end{array}
\]
\end{small}
\noindent which is denoted as:
\begin{small}
\[
\begin{array}{ll}
\intros \| \sum_{t_1} 
\denote{R_1} \; \context \fst{t_1} \times \denote{R_2} \; \context \; \snd{t_1} \times{} \\
\qquad \quad (\denote{c_2} \; \fst{t_1} = \denote{c_3} \; \snd{t_1} ) \times {} \\
\qquad \quad (\denote{c_1} \; \fst{t_1} = \tuple) \times \| \qquad \equiv \\
\intros \| \sum_{t_1} \denote{R_1} \; \context \; \fst{\fst{t_1}} \times \denote{R_1} \; \context \; \snd{\fst{t_1}} \times \denote{R_2} \; \context \; \snd{t_1} \times{} \\
\qquad \quad (\denote{c_1} \; \fst{\fst{t_1}} = \denote{c_1} \; \snd{\fst{t_1}}) \times (\denote{c_2} \; \fst{\fst{t_1}} = \denote{c_3} \; \snd{t_1}) \times{} \\
\qquad \quad (\denote{c_1} \fst{\fst{t_1}} = t) \| 
\end{array}
\]
\end{small}

The decision procedure turns this goal into a bi-implication, which it proves by cases. 
For the $\rightarrow$ case, the decision procedure destructs the available $\Sigma$ witness into
tuple $t_x$ from $R_1$ and $t_y$ from $R_2$ and tries all instantiations of $t_1$ using these tuples. The instantiation $t_1 = ((t_x, t_x), t_y)$ allows the procedure to complete the proof after rewriting all equalities. 
For the $\leftarrow$ case, the available tuples are $t_x$ from $R_1$, $t_y$ from $R_1$, and $t_z$ from $R_2$. 
The instantiation $t_1 = (t_x, t_z)$ allows the procedure to complete the proof after rewriting all equalities. This assignment is visualized in Figure~\ref{fig:mappings}.


\section{Related Work}
\label{sec:related}

\subsection{Query Rewriting}
Query rewriting based on equivalence rules is an essential part of modern query optimizers. Rewrite rules are either fired by a forward chaining rule engine in a Starburst optimizer framework \cite{HaasFLP89SIGMOD,PiraheshHH92SIGMOD}, or are used universally to represent the search space in Exodus \cite{GraefeD87SIGMOD} and its successors, including Volcano \cite{GraefeM93ICDE} and Cascades \cite{Graefe95aDEB}. 

Using \sys, we formally prove a series of rewrite rules from the database literature. Those rules include basic algebraic rewrites such as selection push down \cite{Ullman89}, rewrite rules using indexes \cite{dbSysBook}, and unnesting aggregates with joins \cite{Muralikrishna92VLDB}. We are able to prove one of the most complicated rewrite rules that is also widely used in practice: magic set rewrites \cite{MumickFPR90SIGMOD, BancilhonMSU86PODS, SeshadriHPLRSSS96SIGMOD}. Magic set rewrites involve many SQL features such as correlated subqueries, aggregation and group by. To our best knowledge, its correctness has not been formally proven before. 

\sys automates proving rewrite rules on decidable fragments of SQL. According to Codd's theorem \cite{Codd72}, relational algebra and relational calculus (formulas of first-order logic on database instances) are equivalent in expressive power. Thus, the equivalence between two SQL queries is in general undecidable \cite{Trakhtenbrot50DANUSSR}. Extensive research has been done to study the complexity of containment and equivalence of fragments of SQL queries under bag semantics and set semantics \cite{ChandraM77STOC, ChaudhuriV92PODS, IoannidisR95TODS,  DBLP:conf/icdt/GeckKNS16, JayramKV06PODS, SagivY80JACM, Meyden92PODS}. We list the results in Figure~\ref{fig:complexity}.

\subsection{SQL Semantics} 
SQL is the de-facto language for relational database systems. Although the SQL language is an ANSI/ISO standard \cite{sql2011}, it is loosely described in English and leads to conflicting interpretations \cite{Date89AW}. Previous related formalizations of various fragments of SQL include relational algebra \cite{TheAliceBook}, comprehension syntax \cite{BunemanLSTW94SIGMOD}, and recursive and non-recursive Datalog \cite{ChaudhuriV92PODS}. These formalisms are not suited for rigorous reasoning about the correctness of real world rewrite rules since they mostly focus exclusively on set semantics. In addition, in order to express rewrite rules in these formalism, non-trivial transformation from SQL are required.

Previous SQL formalizations in proof systems include \cite{MalechaMSW10POPL, BenzakenCD14ESOP, VeanesGHT09ICFEM, VeanesTH10LAPR16}. In \cite{VeanesGHT09ICFEM, VeanesTH10LAPR16}, SQL semantics are encoded in the Z3 SMT solver for test generation.  In \cite{MalechaMSW10POPL}, an end to end verified prototype database system is implemented in Coq. In \cite{BenzakenCD14ESOP}, a relational data model and relational algebra are implemented in Coq. Compared with \cite{MalechaMSW10POPL, BenzakenCD14ESOP}, \sem covers all important SQL feature such as bags, aggregation, group by, indexes, and correlated subqueries. As a result, we are able to express a wide range of rewrite rules. Unlike \cite{MalechaMSW10POPL}, we did not build an end to end formally verified database system. 

In \sem, SQL features like aggregation on group by and indexes are supported through syntactic rewrites. Rewriting aggregation on group by using aggregation on relations and correlated subqueries is based on \cite{BunemanLSTW94SIGMOD}. We use logical relation to represent indexes in \sem. This was firstly proposed by Tastalos et al \cite{TsatalosSI94}. 

\subsection{Related Formal Semantics in Proof Systems} 

In the past decades, a number of formal semantics in different application domains were developed using proof systems for software verification. The CompCert compiler \cite{Leroy09JACM} specifies the semantics of a subset of C in Coq and provides machine checkable proofs of the correctness of compilation. HALO denotes Haskell to first-order logic for static verification of contracts \cite{VytiniotisJCR13POPL}.  
Bagpipe \cite{WeitzWTEKT2016:TR} developed formal semantics for the Border Gateway Protocol (BGP) to check the correctness of BGP configurations. SEL4 \cite{SEL4} formally specifies the functional correctness of an OS kernel in Isabelle/HOL and developed a verified OS kernel. FSCQ \cite{FSCQ} builds a crash safe file system using an encoding of crash Hoare logic in Coq. With formal semantics in proof systems, there are more verified system developed such as Verdi \cite{WilcoxWPTWEA15PLDI}, Verve \cite{YangH10PLDI}, Bedrock \cite{Chlipala13} and Ironclad \cite{HawblitzelHLNPZZ14OSDI}.


\section{Discussion}
\label{sec:discussion}

{\bf Limitations} Our system does currently not support three SQL
features: NULL's with their associated three-valued-logic, outer
joins, and windows functions.  However, all can be expressed in \sem,
at the cost of some added complexity, as we explain now.

When any argument to an expression is NULL, then the expression's
output is NULL; this feature can easily be supported by modifying the
external operators.  When an argument of a comparison predicate is
NULL, then the resulting predicate has value \texttt{unknown}, and SQL
uses three valued logic to compute predicates: it defines
$0 = \texttt{false}, 1/2 = \texttt{unknown}, 1=\texttt{true}$, and the
logical operators $x \texttt{ and } y = \min(x,y)$,
$x \texttt{ or } y = \max(x,y)$, $\texttt{not}(x) = 1-x$; a
select-from-where query returns all tuples for which the
where-predicate evaluates to \texttt{true} (i.e. not \texttt{false} or
\texttt{unknown}).  As a consequence, the law of excluded middle
fails, for example the query:
$$\SelectFromWhere{\Star}{R}{a=5 \texttt{ or } a \neq 5}$$
is not equivalent to $\SelectFrom{\Star}{R}$.  This, too, could be
currently expressed \sem\ by encoding the predicates as external
functions that implement the 3-valued logic.  However, by doing so one
hides from the rewrite rules the equality predicate, which plays a key
role in joins.  In future versions, we plan to offer native support
for NULL's, to simplify the task of proving rewrite rules over
relations with NULLs.

Both outer joins and windows functions are directly expressible in
\sem.  For example, a left outer join of two relations $R(a,b)$,
$S(b,c)$ can be expressed by first joining $R$ and $S$ on $b$, and
union-ing the result with
\[
\begin{array}{ll}
\SelectFrom{R.*,\texttt{NULL}}{S} \texttt{ EXCEPT } \\
\SelectFromWhere{R.*,\texttt{NULL}}{S}{R.b=S.b} 
\end{array}
\] 
A direct
implementation in \sem\ would basically have to follow the same
definition of left outer joins.

{\bf Finite v.s. Infinite Semantics} Recall that our semantics extends
the standard bag semantics of SQL in two ways: we allow a relation to
have infinitely many distinct elements, and we allow each element to
have an infinite multiplicity.  To the best of our knowledge, our
system is the first that interprets SQL over infinite relations.  This
has two consequences. First, our system cannot check the equivalence
of two SQL expressions that return the same results on all finite
relations, but differ on some infinite relations.  It is well-known
that there exists First Order sentences, called {\em infinity axioms},
that do not admit any finite model, but admit infinite models.  For
example \cite[pp.307]{DBLP:books/sp/BorgerGG1997} the sentence
$\varphi \equiv \forall x \exists y \forall z(\neg R(x,x) \wedge
R(x,y) \wedge (R(y,z) \rightarrow R(x,z)))$ is an infinity axiom.  It
is possible to write a SQL query that checks $\varphi$, then returns
the empty set if $\varphi$ is false, or returns a set consisting of a
single value (say, 1) if $\varphi$ is true: call this query $Q_1$.
Call $Q_2$ the query \SelectFromWhere{\texttt{DISTINCT} 1}{R}{2=3}.
Then $Q_1=Q_2$ over all finite relations, but $Q_1 \neq Q_2$ not over
infinite relations.  Thus, one possible disadvantage of our semantics
is that we cannot prove equivalence of queries that encode infinity
axioms.  However, none of the optimization rules that we found in the
literature, and discussed in this paper, encode an infinity axiom.
Hence we argue that, for practical purposes, extending the semantics
to infinite relations is a small price to pay for the added simplicity
of the equivalence proofs.  Second, by generalizing SQL queries to
both finite and infinite relations we make our system theoretically
complete: if two queries are equivalent then, by G\"odel's
completeness theorem, there exists a proof of their equivalence.
Finding the proof is undecidable (it is recursively enumerable, r.e.):
our system does not search for the proof, instead the user has to find
it, and our system will verify it.  Contrast this with a system whose
semantics is based on finite relations: such a system cannot have a
complete proof system for SQL query equivalence.  Indeed, if such a
complete proof system existed, then SQL query equivalence would be
r.e. (since we can enumerate all proofs and search for a proof of
$Q_1=Q_2$), and therefore equivalence would be decidable (since it is
also co-r.e., because we can enumerate all finite relations, searching
for an input s.t. $Q_1 \neq Q_2$).  However, by Trakthenbrot's, query
equivalence is undecidable.  Recall that Trakthenbrot's
theorem~\cite{Trakhtenbrot50DANUSSR,DBLP:books/sp/Libkin04} states
that the problem {\em given an FO sentence $\varphi$, check if
  $\varphi$ has a finite model} is undecidable.  We can reduce this
problem to query equivalence by defining $Q_1$ to be a query that
checks checks $\varphi$ and returns the empty set if $\varphi$ is
false, or returns some non-empty set if $\varphi$ is true, and
defining $Q_2$ to be the query that always returns the empty set (as
above), then checking $Q_1 \equiv Q_2$).  Thus, by extending our
semantics to infinite relations we guarantee that, whenever two
queries are equivalent, there exists a proof of their equivalence.


\section{Conclusion}
\label{sec:conclusion}

We have described \sys, a system for proving equivalence of SQL
rewrite rules.  In support of \sys, we defined a formal language \sem,
following closely SQL's syntax.  Our semantics extends that of SQL
from finite relations to infinite relations, and uses univalent types
from Homotopy Type Theory to represent and prove equality of cardinal
numbers (finite and infinite).  We have demonstrated the power and
flexibility of \sys by proving the correctness of several powerful
optimization rules found in the database literature.


\newpage




\bibliographystyle{abbrvnat}
\bibliography{paper}




\end{document}